\documentclass[11pt]{article}
\usepackage[T1]{fontenc}
\usepackage{amsfonts,amsmath,amsthm,amssymb,dsfont,mathtools}
\usepackage[margin=1in]{geometry}
\usepackage{fullpage}
\usepackage{enumitem}
\usepackage[linesnumbered,boxed,ruled,vlined]{algorithm2e}
\usepackage[colorlinks,citecolor=blue,linkcolor=blue,urlcolor=red]{hyperref}
\usepackage[capitalize,nameinlink]{cleveref}
\usepackage{url}
\usepackage{thmtools, thm-restate}
\usepackage{xcolor}
\usepackage{graphicx}
\usepackage{mathdots}
\usepackage[normalem]{ulem}
\usepackage{xspace}
\xspaceaddexceptions{]\}}
\usepackage{comment}
\usepackage{tabu}
\usepackage{framed}
\usepackage{float,wrapfig}
\usepackage{subfigure}
\usepackage{tikz}
\usepackage{thm-restate}
\usetikzlibrary{arrows.meta}

\nocite{*}

\crefname{lemma}{Lemma}{Lemmas}
\crefname{claim}{Claim}{Claims}
\newtheoremstyle{dotless}{}{}{\itshape}{}{\bfseries}{}{ }{}
\theoremstyle{dotless}
\newtheorem{theorem}{Theorem}[section]
\newtheorem{lemma}[theorem]{Lemma}
\newtheorem{proposition}[theorem]{Proposition}
\newtheorem{cor}[theorem]{Corollary}
\newtheorem{remark}[theorem]{Remark}
\newtheorem{observation}[theorem]{Observation}

\newtheoremstyle{dotlessdef}{}{}{\normalfont}{}{\bfseries}{}{ }{}
\theoremstyle{dotlessdef}
\newtheorem{definition}[theorem]{Definition}

\newtheorem{claim}[theorem]{Claim}

\newcommand{\E}{\operatorname*{\mathbb{E}}}

\newcommand{\lpr}[1]{\left(#1\right)}
\newcommand{\lbr}[1]{\left[#1\right]}

\newcommand{\Btwoin}{B_{\mathrm{in}}^{(2)}}
\newcommand{\Btwoout}{B_{\mathrm{out}}^{(2)}}

\newcommand{\Bfourout}{B_{\mathrm{out}}^{(4)}}
\newcommand{\Rin}{R_{1,\mathrm{in}}}
\newcommand{\Rout}{R_{1,\mathrm{out}}}
\newcommand{\Rtwoin}{R_{2,\mathrm{in}}}

\DeclareMathOperator{\wt}{wt}

\newcommand{\poly}{\operatorname{\mathrm{poly}}}
\newcommand{\polylog}{\poly\log}

\newcommand{\lr}{\leftrightharpoons}

\usepackage[
backend=biber,
style=alphabetic,
maxnames=99,
maxalphanames=5
]{biblatex}
\renewbibmacro{in:}{%
  \ifentrytype{article}{}{\printtext{\bibstring{in}\intitlepunct}}}

\usepackage[mathlines]{lineno}
\usepackage{etoolbox}
\newcommand*\linenomathpatch[1]{%
  \cspreto{#1}{\linenomath}%
  \cspreto{#1*}{\linenomath}%
  \csappto{end#1}{\endlinenomath}%
  \csappto{end#1*}{\endlinenomath}%
}
\linenomathpatch{equation}
\linenomathpatch{gather}
\linenomathpatch{multline}
\linenomathpatch{align}
\linenomathpatch{alignat}
\linenomathpatch{flalign}

\newcommand{\eps}{\varepsilon}

\newcommand{\mmid}{\,\Big|\,}

\title{Improved Roundtrip Spanners, Emulators, and \\ Directed Girth Approximation}
\author{Alina Harbuzova\thanks{hadought@mit.edu} \\ MIT \and Ce Jin\thanks{cejin@mit.edu, Partially supported by NSF Grant CCF-2129139.} \\ MIT\and Virginia Vassilevska Williams\thanks{virgi@mit.edu, Supported by NSF Grants CCF-2129139 and CCF-2330048 and BSF Grant 2020356.}\\ MIT \and Zixuan Xu\thanks{zixuanxu@mit.edu}\\ MIT}

\date{}

\begin{document}
	\setcounter{page}{0} \clearpage
	\maketitle
	\thispagestyle{empty}
	\begin{abstract}
Roundtrip spanners are the analog of spanners in directed graphs, where the roundtrip metric is used as a notion of distance. Recent works have shown existential results of roundtrip spanners nearly matching the undirected case, but the time complexity for constructing roundtrip spanners is still widely open.

This paper focuses on developing fast algorithms for roundtrip spanners and related problems. 
		For any $n$-vertex directed graph $G$ with $m$ edges (with non-negative edge weights), our results are as follows:
		\begin{itemize}
			\item \textbf{3-roundtrip spanner faster than APSP:} We give an $\tilde O(m\sqrt{n})$-time algorithm that constructs a roundtrip spanner of stretch $3$ and optimal size $O(n^{3/2})$.
		Previous constructions of roundtrip spanners of the same size either required $\Omega(nm)$ time [Roditty, Thorup, Zwick SODA'02; Cen, Duan, Gu ICALP'20], or had worse stretch $4$ [Chechik and Lifshitz SODA'21].
		\item \textbf{Optimal roundtrip emulator in dense graphs:} For integer $k\ge 3$, we give an $O(kn^2\log n)$-time algorithm that constructs a roundtrip  \emph{emulator} of stretch $(2k-1)$ and size $O(kn^{1+1/k})$, which is optimal for constant $k$ under Erd\H{o}s' girth conjecture.	
			Previous work of [Thorup and Zwick STOC'01] implied a roundtrip emulator of the same size and stretch, but it required $\Omega(nm)$ construction time.
		 Our improved running time is near-optimal for dense graphs.
			\item \textbf{Faster girth approximation in sparse graphs:} We give an $\tilde O(mn^{1/3})$-time algorithm that $4$-approximates the girth of a directed graph. This can be compared with the previous $2$-approximation algorithm in $\tilde O(n^2, m\sqrt{n})$ time by [Chechik and Lifshitz SODA'21]. 
   In sparse graphs, our algorithm achieves better running time at the cost of a larger approximation ratio.
		\end{itemize}
	\end{abstract}
	\newpage
 
\maketitle

\section{Introduction}
A $t$-spanner of a graph is a subgraph that approximates all pairwise distances within a factor of $t$. Spanners are useful in many applications since they can be significantly sparser than the graphs they represent, yet are still a good representation of the shortest paths metric. As many algorithms are much faster on sparse graphs, running such algorithms on a spanner rather than the graph itself can be significantly more efficient, with only a slight loss in approximation quality.

For undirected graphs, the spanner question is very well understood. It is known that for all integers $k\geq 2$, every $n$-vertex undirected (weighted) graph contains a $(2k-1)$-spanner on $O(n^{1+1/k})$ edges \cite{AlthoferDDJS93} and this is optimal under Erd\H{o}s' girth conjecture \cite{TZ01}. 

For directed graphs, however, there can be no non-trivial spanners under the usual shortest paths metric: consider for instance a complete bipartite graph,  with edges directed from one partition to the other. Omitting a single edge $(u,v)$ would cause the distance $d(u,v)$ to go from $1$ to $\infty$.

Nevertheless, one can define a notion of a spanner in directed graphs based on the {\em roundtrip} metric defined by Cowen and Wagner \cite{CowenW04}: $d(u\lr v)=d(u,v)+d(v,u)$. 
A roundtrip $t$-spanner of a directed graph is a subgraph that preserves all pairwise roundtrip distances within a factor of $t$.

Cen, Duan and Gu \cite{duan} showed that basically the same existential results are possible for roundtrip spanners as in undirected graphs: for every integer $k\geq 2$ every $n$-vertex directed graph contains a $(2k-1)$-roundtrip spanner on $O(kn^{1+1/k}\log n)$ edges. For the special case of $k=2$, it was known earlier that every $n$-vertex graph contains a $3$-roundtrip spanner on $O(n\sqrt{n})$ edges \cite{RodittyTZ08}.

The known results on algorithms for constructing spanners and roundtrip spanners differ drastically however. Baswana and Sen \cite{BaswanaS07} presented a randomized linear time algorithm for computing an $O(kn^{1+1/k})$-edge $(2k-1)$-spanner of any $n$-vertex weighted graph (which was later derandomized \cite{RodittyTZ05}).
Meanwhile, the algorithms for constructing roundtrip spanners are much slower.

The first construction of roundtrip spanners was given by Roditty, Thorup and Zwick in \cite{RodittyTZ08}, where they gave the construction of $(2k+\eps)$-roundtrip spanners on $\tilde{O}((k^2/\eps) n^{1+1/k})$ edges for any graph with edge weights bounded by $\poly n$ ($\log nW$ dependence in the size otherwise) in $O(mn)$ time. Later, Zhu and Lam~\cite{DBLP:journals/ipl/ZhuL18} derandomized this construction and improved the sparsity of the spanner to contain $\tilde{O}((k/\eps) n^{1+1/k})$ edges.  Most recently, Chechik and Lifshitz constructed a $4$-roundtrip spanner on $O(n^{3/2})$ edges in $\Tilde{O}(n^2)$ time. All currently known results on constructions of roundtrip spanners are summarized in \cref{tab:spanner-results}.

Notice that for all cases with running time faster than $mn$, the stretch is suboptimal for the used sparsity. This motivates the following:
\begin{center}{\em
Question: What is the best construction time for roundtrip spanners of optimal stretch-sparsity tradeoff?}\end{center}

Alongside the construction of roundtrip spanners, another closely related problem is approximating the girth (i.e. the length of the shortest cycle) in directed graphs. The first nontrivial algorithm is by Pachocki, Roditty, Sidford, Tov, Vassilevska Williams~\cite{PachockiRSTW18}, who gave an $O(k\log n)$ approximation algorithm running in $\Tilde{O}(mn^{1/k})$ time. Further improvements by \cite{chechikliu20,DalirrooyfardW20} followed. Most recently, Chechik and Lifshitz~\cite{CL21} obtained a $2$-approximation in $\Tilde{O}(\min\{n^2, m\sqrt{n}\})$ time, which is optimal for dense graphs. The current known results are summarized in \cref{tab:girth-results}.

While the $2$-approximation result is optimal for dense graphs, and while a $2-\eps$-approximation is (conditionally) impossible in $O((mn)^{1-\delta})$ time \cite{DalirrooyfardW20}, it is unclear what other approximations ($2.5? ~3?$) are possible with faster algorithms. 
This motivates the following question:
\begin{center}{\em
Question: what is the best running time-approximation tradeoff for the girth of directed graphs?}\end{center}

\subsection{Our Results}

Throughout this paper, we consider directed graphs on $n$ vertices and $m$ edges with non-negative edge weights.
We use $\tilde O(\cdot )$ to hide $\polylog(n)$ factors. 
All our algorithms are Las Vegas randomized. 

\begin{restatable}{theorem}{3spanner-main}
\label{thm:3roundtrip-main}
There is a randomized algorithm that computes a 3-roundtrip spanner of $O(n^{3/2})$ size in $\tilde O(m\sqrt{n})$ time. 
\end{restatable}

This can be compared with the 4-roundtrip spanner of $O(n^{3/2})$ size constructable in $O(n^2\log n)$ time from \cite{CL21}. 

Alongside spanners, another important object of study are {\em emulators}: sparse graphs that approximate all pairwise distances; the difference here is that emulators are not required to be subgraphs, and can be weighted even if the original graph was unweighted. Similar to roundtrip spanners being analogs of spanners in directed graphs, we consider roundtrip emulators which are the analogs of emulators in directed graphs. While emulators are very well studied in undirected graphs \cite{AingworthCIM99,DorHSZ96,Woodruff06,pettie2009low,BaswanaSKT10,BW15,BV16,AB,Huang2018LowerBO,Lu2021BetterLB,koganP23},
 the authors are not aware of any results, for the roundtrip metric. The only known construction of roundtrip emulators is implied from using the roundtrip metric in Thorup-Zwick's distance oracle in \cite{TZ01}, which has $(2k-1)$-stretch and $O(kn^{1+1/k})$ edges but requires $\Tilde{O}(mn)$ construction time.

We obtain a very fast algorithm that constructs essentially optimal roundtrip emulators (up to the Erd\H{o}s girth conjecture).

\begin{restatable}{theorem}{emulator-main}
\label{thm:emulator-main}
For integers $k\ge 3$, there is a randomized algorithm that computes a $(2k-1)$-roundtrip emulator of $O(kn^{1+1/k})$ size  in $O(kn^{2}\log n)$ time. 
\end{restatable}

While the result is only for roundtrip emulators, rather than spanners, it achieves a much faster running time than any result on roundtrip spanners with optimal approximation-size tradeoff. This is the first algorithm that achieves a sub-$mn$ running time for the problem.

We next focus on the closely related question of girth approximation. We prove:

\begin{restatable}{theorem}{4girth-main}
\label{thm:4girth-main}
There is a randomized algorithm that computes a 4-multiplicative approximation of the girth of a directed graph in $\tilde O(mn^{1/3})$ time. 
\end{restatable}

Let us compare with the previous known directed girth approximation algorithms.
Compared with the 2-approximation in $\tilde O(n^2,m\sqrt{n})$ time from \cite{CL21}, Theorem \ref{thm:4girth-main} achieves a better running time for $m\leq o(n^{5/3})$ while raising the approximation ratio to $4$. Dalirrooyfard and Vassilevska W. \cite{DalirrooyfardW20} gave for every constant $\eps>0$, a $(4+\eps)$-approximation algorithm running in $\tilde{O}(mn^{\sqrt{2}-1})$ time. Our algorithm removes the $\eps$ from the approximation factor and further improves the running time.

\subsection{Paper organization}
After introducing useful notations and terminologies in \cref{sec:prelim}, we give a high level overview of our techniques in \cref{sec:tech}. Then, in \cref{sec:3-spanner} we describe our $3$-roundtrip spanner algorithm (\cref{thm:3roundtrip-main}). In \cref{sec:emulator} we describe our roundtrip emulator algorithm. In \cref{sec:4-approx} we describe our girth approximation algorithm. We conclude with a few open questions in \cref{sec:conclusion}.

\begin{table}[]
    \centering
    {\renewcommand{\arraystretch}{1.2}
    \begin{tabular}{|m{6cm}|c|c|c|}
        \hline
         \textbf{Citation} & \textbf{Stretch} & \textbf{Sparsity} & \textbf{Time}  \\ \hline 
         Roditty, Thorup, Zwick~\cite{RodittyTZ08} \textsuperscript{$\triangle$} & $2k+\eps$ & $\Tilde{O}\lpr{\frac{k^2}{\eps}n^{1+1/k}}$ & $O(mn)$ \\ \hline
         
         Pachocki, Roditty, Sidford, Tov, Vassilevska W.~\cite{PachockiRSTW18} & $O(k\log n)$ & $\Tilde{O}(n^{1+1/k})$ & $\Tilde{O}(mn^{1/k})$ \\ \hline
         
         Chechik, Liu, Rotem, Sidford~\cite{chechikliu20} & $O(k\log k)$ & $\Tilde{O}(n^{1+1/k})$ & $\Tilde{O}(m^{1+1/k})$\\ \hline
         
         Cen, Duan, Gu~\cite{duan} & $2k-1$ & $\Tilde{O}(kn^{1+1/k})$ & $\Tilde{O}(mn\log W)$\\ \hline

         Chechik, Liu, Rotem, Sidford~\cite{chechikliu20} \textsuperscript{$\triangle$} & $8+\eps$ & $\Tilde{O}(n^{3/2}/\eps^2)$ & $\Tilde{O}(m\sqrt{n})$\\ \hline

         Dalirrooyfard and Vassilevska W.~\cite{DalirrooyfardW20} \textsuperscript{$\triangle$}& $5+\eps$ & $\Tilde{O}(n^{3/2}/\eps^2)$ & $\Tilde{O}(m\sqrt{n})$\\ \hline

         Chechik and Lifshitz~\cite{CL21} & $4$ & $O(n^{3/2})$ & $\Tilde{O}(n^2)$\\ \hline

         \textbf{New} & $3$ & $O(n^{3/2})$ & $\Tilde{O}(m\sqrt{n})$\\ \hline
    \end{tabular}
    }
    \caption{Known results on constructions of roundtrip spanners on a weight directed graph on $n$ vertices and $m$ edges with edge weight bounded by $W$. Results marked with $\triangle$ are subsumed by other results.}
    \label{tab:spanner-results}
\end{table}

\begin{table}[]
    \centering
    {\renewcommand{\arraystretch}{1.2}
    \begin{tabular}{|m{5cm}|c|c|}
        \hline
         \textbf{Citation} & \textbf{Approximation Factor}  & \textbf{Time}  \\ \hline

         Pachocki, Roditty, Sidford, Tov, Vassilevska W. \cite{PachockiRSTW18} & $O(k\log n)$ &  $\Tilde{O}(mn^{1/k})$ \\ \hline

         Chechik, Liu, Rotem, Sidford \cite{chechikliu20} \textsuperscript{$\triangle$} & $3$ & $\Tilde{O}(m\sqrt{n})$\\ \hline
         
         Chechik, Liu, Rotem, Sidford \cite{chechikliu20} & $O(k\log k)$ &  $\Tilde{O}(m^{1+1/k})$\\ \hline

         Dalirrooyfard and Vassilevska W. \cite{DalirrooyfardW20} \textsuperscript{$\triangle$} & $4+\eps$ & $\Tilde{O}(mn^{\sqrt{2}-1})$\\ \hline

         Dalirrooyfard and Vassilevska W. \cite{DalirrooyfardW20} \textsuperscript{$\triangle$}& $2+\eps$ & $\Tilde{O}(m\sqrt{n})$ \\ \hline

         Dalirrooyfard and Vassilevska W. \cite{DalirrooyfardW20} \textsuperscript{$\triangle$}& $2$ & $\Tilde{O}(mn^{3/4})$ (unweighted)\\ \hline

         Chechik and Lifshitz \cite{CL21} & $2$ & $\Tilde{O}(\min\{n^2,m\sqrt{n}\})$\\ \hline

         \textbf{New} & $4$ & $\Tilde{O}(mn^{1/3})$\\ \hline
    \end{tabular}
    }
    \caption{Known results on  girth approximation on a weight directed graph on $n$ vertices and $m$ edges with edge weight bounded by $W$. Results marked with $\triangle$ are subsumed by other results. }
    \label{tab:girth-results}
\end{table}
\section{Preliminaries}\label{sec:prelim}

We use $\tilde O(\cdot )$ to hide $\polylog(n)$ factors, where $n$ is the number of vertices in the input graph. 

In this paper, the input graph $G = (V,E)$ is always a weighted directed graph with vertex set $V$ of size $|V| = n$ and edge set $E$ of size $|E| = m$ with non-negative edge weights.
Without loss of generality, we assume $G$ does not have parallel edges.
We use $\wt(u,v)$ to denote the weight of the directed edge $(u,v)\in E$.
 For any two vertices $u,v\in V$, we use $d_G(u,v)$ to denote the distance (length of the shortest path) from $u$ to $v$ in $G$, and we use $d_G(u\lr v) := d_G(u,v) +d_G(v,u)$ to denote the \emph{roundtrip distance} between $u$ and $v$. When the context is clear, we simply use $d(u,v)$ and $d(u\lr v)$. For a subset of vertices $W\subseteq V$, we use $G[W]$ to denote the subgraph of $G$ induced by the vertex set $W$.

 The \emph{girth} of $G$ is the length (total edge weight) of the shortest cycle in $G$. 
We say a graph $H=(V,E')$  is an $\alpha$-\emph{roundtrip emulator} of graph $G=(V,E)$, if for every two vertices $u,v\in V$ it holds that $d_G(u \lr v) \le d_{H}(u\lr v) \le \alpha \cdot d_G(u\lr v)$. 
Furthermore, if $H$ is a subgraph of $G$, we say $H$ is an $\alpha$-\emph{roundtrip spanner} of $G$.

 Without loss of generality, we may assume $G$ is strongly-connected, since otherwise we can run the algorithm for girth approximation (or roundtrip spanner/emulator) on each strongly-connected component.
In addition, we may assume the maximum degree of $G$ is bounded by $O(m/n)$. This is due to the following regularization lemma shown in \cite{chechikliu20}. This assumption will be used in \cref{sec:4-approx}.

\begin{lemma}[Regularization \cite{chechikliu20}]
    \label{lem:regularize}
Given a directed weighted graph $G = (V,E)$ on $n$ vertices and $m$ edges, one can construct a graph $H$ on $O(n)$ vertices and $O(m)$ edges with non-negative edge weights and maximum degree $O(m/n)$ in $O(m)$ time such that all of the following holds:
\begin{enumerate}
    \item All roundtrip distances between pairs of vertices in $G$ are the same in $H$ as in $G$.
    \item  Given a cycle in $H$, one can find a cycle of the same length in $G$ in $O(m)$ time.
    \item  Given a subgraph $H'$ in $H$, one can find in $O(m)$ time a subgraph $G'$ of $G$ such that $|E(G')|\le |E(H')|$ and the roundtrip distances in $G'$ are the same as in $H'$.
\end{enumerate}
\end{lemma}

In our algorithms, we often use Dijkstra's algorithm to compute single-source distances. 
On a weighted directed graph $G = (V,E)$, we use \emph{out-Dijkstra} from a source $s\in V$ to refer to Dijkstra algorithm computing distances $d(s,\cdot)$ from $s$, and use \emph{in-Dijkstra} from $s$ to refer to Dijkstra algorithm computing distances $d(\cdot,s)$ into $s$.

\section{Technical Overview}
\label{sec:tech}
\subsection{Previous Work}\label{subsec:previous-work}

Throughout this paper, our techniques are based on the following key observation introduced in \cite{CL21}.

\begin{lemma}[Key Observation \cite{CL21}]\label{lem:observation}
    Let $G = (V, E)$ be a weighted directed graph with nonnegative edge weights. For  vertices $u,v,r\in V$, if 
    \begin{equation}\label{eq:obs-condition}
        2\cdot d(v,r) + d(r,u) \le 2\cdot d(v,u) + d(u,r),
    \end{equation}
    then
    \[d(u\lr r) \le 2\cdot d(u\lr v).\]
\end{lemma}

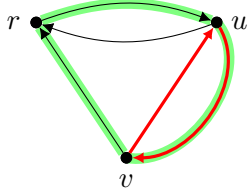
\begin{figure}[H]
    \centering
    \tikzset{vtx/.style = {circle, draw, fill=black, inner sep=0pt, minimum width=4pt},>={Latex[width=1.5mm,length=1.5mm]}}
    \begin{tikzpicture}

    \node[vtx] (v) at (0,0) {};
    \node[vtx] (u) at (1.2,1.8) {};
    \node[vtx] (w) at (-1.2,1.8) {};

    \node at (0,-0.3) {$v$};
    \node at (1.5, 1.8) {$u$};
    \node at (-1.5, 1.8) {$r$};

    \draw[green!50, line width=4pt](v) -- (w);
    \draw[green!50, line width=4pt] (w) to[out=20, in=160] (u);
    \draw[green!50, line width=4pt] (1.2, 1.8) to[out=-50,in=-10] (0,0);

    \node[vtx]  at (0,0) {};
    \node[vtx]  at (1.2,1.8) {};
    \node[vtx]  at (-1.2,1.8) {};

    \draw[->,red,very thick] (v) -- (u);
    \draw[->] (v) -- (w);
    \draw[->] (u) to[out=-160, in=-20] (w);
    \draw[->] (w) to[out=20, in=160] (u);
    \draw[->, red, very thick] (u) to[out=-60,in=0] (v);
    \end{tikzpicture}
    \caption{A illustration of \cref{lem:observation} with 
    $u,v,r\in V$ satisfying \cref{eq:obs-condition}.
 The red cycle $u\leadsto v\leadsto u$ can be $2$-approximated by the cycle $u\leadsto v \leadsto r\leadsto u$ highlighted green.}
    \label{fig:k-approx-lem}
\end{figure}

An important property of the above observation is that \cref{eq:obs-condition} is symmetric with respect to the roles of $u$ and $r$. This symmetry is crucial to the analysis of the applications of \cref{lem:observation} in the previous work \cite{CL21} as well as in our new algorithms, so we first describe it in more details as follows.

\paragraph{The Symmetry Argument}

Consider the following routine that sparsifies a graph $G = (V, E)$ on $n$ vertices using a random sample $S\subseteq V$. For every vertex $v\in V$, we check for every vertex $s\in S\cap N(v)$ and $u\in N(v)$ where $N(v)$ denotes the out neighborhood of $v$, if $2d(v,s) + d(s,u)\le 2d(v,u)+d(u,s)$ then remove the edge $(v,u)$. We say that we use the set $S$ as \emph{eliminators} to perform the sparsification since we are comparing the distance $d(v,u)$ using the distance information involving $s\in S$. 

For any two neighbors $u,u'\in N(v)$ (possible $u = u'$), notice that the condition involving $v,u,u'$ compares the distances $2d(v,u) + d(u,u')$ against $2d(v,u')+d(u',u)$, which is the same as if we switch the roles of $u$ and $u'$. This means that either $u$ eliminates $u'$ or $u'$ eliminates $u$. (We say ``$u$ eliminates $u'$'' meaning that, if $u\in S$, then the edge $(v,u')$ will be removed, namely $u'$ is eliminated from $N(v)$.) So given a random $u\in N(v)$, in expectation half of the pairs $(u,u')$ falls in the case where $u$ can eliminate $u'$ and additionally $u$ can eliminate $u$ itself. Thus, if $N(v)\cap S\ne \varnothing$, then in expectation the procedure will remove at least $|N(v)|/2$ edges. This implies that the graph sparsification can effectively remove a constant fraction of the edges adjacent to the vertices with high out-degree. More specifically, since on expectation, the sample $S$ can hit vertex sets with size $\Omega(n/|S|)$, this procedure can remove a constant fraction of the outgoing edges adjacent to vertices with degree $\Omega(n/|S|)$. So if we repeat this process $\Theta(\log n)$ rounds, on expectation we can reduce the out-degree of every vertex to at most $O(n/|S|)$.

\paragraph{Applications of the key observation}

Now we are ready to explain how the above \cref{lem:observation} is useful for constructing roundtrip spanners and approximating directed cycles. 
\begin{enumerate}
    \item \textbf{Girth approximation: reduce search space.} Suppose we can take a small random subset of vertices $S\subseteq V$ and for each vertex $s\in S$ and set the current girth estimate as the length of the shortest cycle passing through any vertex in $S$. Then if $r\in S$, \cref{lem:observation} shows that we no longer have to consider the shortest cycle passing through $v$ and $u$ satisfying \cref{eq:obs-condition}. This is because the shortest cycle passing through $u$ and $v$ can already be $2$-approximated by the shortest cycle passing through $r$. Then if we want to search for cycles passing through $v$ that cannot be $2$-approximated, we would not need to consider the vertex $u$. Thus, using the sample $S$, we can compute a pruned vertex set $B(v)\subseteq V$ that contains all the vertices $u\in V$ such that \cref{eq:obs-condition} does not hold for any $r\in S$. By the symmetry argument, each sample that hits the set $B(v)$ can reduce the size of $B(v)$ by a constant fraction. So over $\Theta(\log n)$ rounds, we can obtain a pruned set of size roughly $O(n/|S|)$. This technique is used in the $2$-approximation in \cite{CL21} and will be used in our algorithm for computing a $4$-approximation of the girth in \cref{sec:4-approx}.

    \item \textbf{Roundtrip spanners: graph sparsification.} Suppose we take a random subset $S\subseteq V$ and add all the in/out shortest path tree from $S$ to our spanner $H$. We apply \cref{lem:observation} to the vertices $u,v,r\in V$ where $u,r$ are out-neighbors of $v$. If $r\in S$, \cref{lem:observation} implies that we can delete the edge $(v,u)$ since the shortest cycle containing $(v,u)$ can be $2$-approximated by the cycle passing through $r$ and $u$, which is already added to the spanner $H$. As explained previously by the symmetry argument, in expectation we can reduce the out-degree of every vertex to roughly $O(n/|S|)$ if we repeat this process for $\Theta(\log n)$ rounds. This technique was used in the construction of $4$-roundtrip spanners in \cite{CL21} and will be used in our construction of $3$-roundtrip spanner in \cref{sec:3-spanner} and our $(2k-1)$-roundtrip emulator in \cref{alg:emulator}.
    
\end{enumerate}

\subsection{Our Techniques}

Our techniques consist of a collection of extensions to the techniques introduced in \cite{CL21}. We now highlight the novel components in each of our algorithms.

\paragraph{$3$-Roundtrip Spanner in $\Tilde{O}(m\sqrt{n})$ Time}
Our algorithm follows from a modification of Chechik and Lifshitz's \cite{CL21} $4$-roundtrip spanner algorithm, which was based on the graph sparsification approach mentioned earlier.  
Our new idea lies in a more careful analysis of the stretch of the spanner: instead of directly bounding the roundtrip distance $d_H(u\lr v)$ between vertices $u,v$ in the spanner $H$ as Chechik and Lifshitz did, we separately bound the one-way distances $d_H(u,v), d_H(v,u)$ and add them up. 
After a slight change in their algorithm (namely,
by computing distances in the original graph rather than in the sparsified graph in each round), this analysis enables us to improve the stretch from $4$ to $3$.

\paragraph{$(2k-1)$-Roundtrip Emulator in $\Tilde{O}(n^2)$ Time}

The celebrated approximate distance oracle result of Thorup and Zwick \cite{TZ01} immediately yields $(2k-1)$-emulators of $O(kn^{1+1/k})$ size for any metric. But a straightforward implementation of their generic algorithm in the roundtrip metric would require computing single source shortest paths from all vertices, in $\tilde O(mn)$ total time.
For the easier case of undirected graphs, \cite{TZ01} reduced the construction time to $O(kmn^{1/k})$, but unfortunately these techniques based on balls and bunches do not yield a speedup in our roundtrip distance setting.  

Our faster roundtrip emulator algorithm combines Thorup and Zwick's technique \cite{TZ01} with the graph sparsification approach of \cite{CL21}.
The intuition is that, since the bottleneck of the generic Thorup-Zwick algorithm lies in computing single source shortest paths, a natural idea is to use \cite{CL21}'s approach to gradually sparsify the graph so that Dijkstra's algorithm can run faster.
More specifically,  recall that the Thorup-Zwick algorithm takes a sequence of nested vertex samples $S_1\subseteq \dots \subseteq S_k = V$ which serve as intermediate points for routing approximate shortest paths. In our case, these vertex samples also play the same role as in the graph sparsification approach described earlier, where short cycles going through these vertex samples can approximate the cycles we care about.  This results in a multi-round algorithm that interleaves graph sparsification steps and running Dijkstra from vertices of $S_i$ (with gradually increasing size) in $\tilde O(n^2)$ total time.
It is not obvious that the $(2k-1)$-stretch of Thorup-Zwick still holds after adding these graph sparsification steps, but it turns out the stretch analysis of Thorup-Zwick fits nicely with the cycle approximation arguments, and with a careful analysis we are still able to show $(2k-1)$ stretch when $k\ge 3$.

For some technical reason related to the sampling argument of Thorup-Zwick, we had to slightly simplify the graph sparsification techniques of \cite{CL21}, in order to avoid an undesirable extra logarithmic factor in the sparsity bound of our roundtrip emulator. See the discussion in \cref{remark:remarkresampling} and the proof of \cref{lem:bunch-size}.

\paragraph{$4$-Approximation of Girth in $\Tilde{O}(mn^{1/3})$ Time}

Our algorithm vastly extends the technique of the $2$-approximate girth algorithm in $\Tilde{O}(\min\{n^2, m\sqrt{n}\})$ time by Chechik and Lifshitz \cite{CL21}.
In the $2$-approximation algorithm, one takes a sample $S$ of size $O(\sqrt{n})$ and uses in/out Dijkstra's to exactly compute the shortest cycle going through every $s\in S$.
Then using $S$ as eliminators, compute for every vertex $v\in V$ a pruned vertex set $B(v)$ of size $O(\sqrt{n})$, and search for short cycles from $v$ on $G[B(v)]$.
A natural attempt to improve the running time is to generalize this framework to multiple levels: take a sequence of vertex samples of increasing sizes $S_1,\dots, S_{k-1},S_k=V$ and compute a sequence of pruned vertex subsets $V=B_1(v),B_2(v),\dots, B_{k}(v)$ of decreasing sizes for every $v$, so that one can run Dijkstra from/to every vertex in $S_i$ on $G[B_i(v)]$ 
in $\tilde O(mn^{1/k})$ time.
However, it is unclear how to do this since one can no longer check the condition \cref{eq:obs-condition} due to not having all the distance information from/to every vertex $s\in S_i$, and thus we cannot compute the sets $B_i(v)$ as desired.

 In this work we are able to implement the above plan for $k=3$, obtaining a $4$-approximation girth algorithm in $\tilde O(mn^{1/3})$ time. 
We deal with the problem of not having enough distance information to compute $B_3(v)$ by using a certain distance underestimate obtained from the distance information from $S_1$, and enforcing a stricter set of requirements on the vertices that we explore,  so that we always have their distance information available.
We also apply more novel structural lemmas about cycle approximation that extend the key observation \cref{lem:observation} of \cite{CL21} in various ways, which may be of independent interest.
As a result, our $4$-approximation algorithm becomes more technical than the previous $2$-approximation algorithm in $\tilde O(m\sqrt{n})$ time.

Here, we highlight the key structural lemma (\cref{lem:filter-lemma}) that enabled us to overcome the above described difficulty. It is illustrated in the following \cref{fig:filter-lem-intuition}, which can be viewed as an extension of \cref{lem:observation} from $3$ vertices to $4$ vertices.
As illustrated, if there exists some vertex $r_2$ that is in a short cycle with $v$ but not in a short cycle with $u$, then we can find some vertex $r_1$ such that the cycle $v \leadsto r_2 \leadsto r_1 \leadsto u \leadsto v$ (highlighted in green) can approximate the shortest cycle passing through $u$ and $v$ (the cycle in red). Then similar to how we can use \cref{lem:observation}, we can ignore the vertex $u$ in our search for the shortest cycle passing through $v$.

\begin{figure}[H]
    \centering
    \tikzset{vtx/.style = {circle, draw, fill=black, inner sep=0pt, minimum width=4pt},>={Latex[width=1.5mm,length=1.5mm]}}
    \begin{tikzpicture}
    \node[vtx] (v) at (0,0) {};
    \node[vtx] (u) at (1.5,1.5){};
    \node[vtx] (r2) at (-1.5,1.5){};
    \node[vtx] (r1) at (0,3){};

    \node at (0,-0.3) {$v$};
    \node at (1.8, 1.5) {$u$};
    \node at (-1.8, 1.5) {$r_2$};
    \node at (0,3.3) {$r_1$};

    \draw[green!50, line width=4pt] (v) to[bend left] (r2);
    \draw[green!50, line width=4pt] (r2) to[bend left] (r1);
    \draw[green!50, line width=4pt] (r1) to[bend left] (u);
    \draw[green!50, line width=4pt] (u) to[bend left] (v);

    \draw[->](r1) to[bend left] (r2);
    \draw[->](r2) to[bend left] (r1);

    \draw[->, red, very thick](v) to[bend left] (u);
    \draw[->, red, very thick](u) to[bend left] (v);

    \draw[->, very thick](v) to[bend left] (r2);
    \draw[->, very thick](r2) to[bend left] (v);

    \draw[->,thick, dashed](r2) to[out=15, in=165] (u);
    \draw[->,thick,dashed](u) to[out=-165, in=-15] (r2);
    
    \draw[->](r1) to[bend left] (u);
    \end{tikzpicture}
    \caption{If there exists a vertex $r_2$ that is in a short cycle with $v$ but not in a short cycle with $u$, then we can find a vertex $r_1$ such that the cycle passing through $v \leadsto r_2 \leadsto r_1 \leadsto u$ (the cycle highlighted in green) can approximate the shortest cycle passing through $u$ and $v$ (the cycle in red).}
    \label{fig:filter-lem-intuition}
\end{figure}
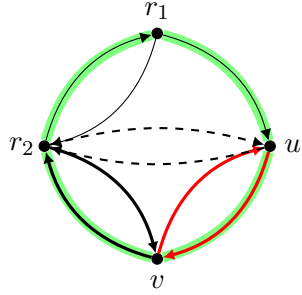

Furthermore, we note that we had to introduced a number of technicalities and a new structural theorem just to implement our proposed generalization for $k = 3$. So it is entirely unclear how to further generalize this approach for $k \ge 4$. Moreover, even if one successfully implements the proposed generalization naively, one would only obtain a $2^{k-1}$-approximation in $\Tilde{O}(mn^{1/k})$ time, which is far from being desirable.

\section{3-Roundtrip Spanner}\label{sec:3-spanner}

In this section, we present our algorithm for constructing a $3$-roundtrip spanner with $O(n^{3/2})$ edges in time $\Tilde{O}(m\sqrt{n})$ (\cref{thm:3roundtrip-main}).
Our algorithm closely follows the previous $\tilde O(n^2)$-time $4$-roundtrip spanner algorithm by Chechik and Lifshitz \cite{CL21}, but we use a more careful analysis to improve the stretch from $4$ to $3$.

\subsection{Algorithm and stretch analysis}

Our algorithm (see pseudocode in \cref{alg:3-spanner}) has a similar structure as in \cite{CL21}: 
We iteratively sample vertex subsets $S_i\subseteq V$ with geometrically increasing expected sizes $\E[|S_i|]$ up to $\sqrt{n}$.
In each iteration $i$, we add the shortest path trees from/to every $s\in S_i$ into the spanner, and sparsify the input graph $G$ using the method of \cite{CL21} based on $S_i$ (\cref{line:8} -- \cref{line:11}). 
Finally, we are able to sparsify the graph to contain only $O(n^{3/2})$ edges in expectation, and we will add these remaining edges to the spanner. 
Over all iterations, we add a total of $2n\cdot O(\sqrt{n}) = O(n^{3/2})$ edges to the spanner, and we only run $O(\sqrt{n})$ instances of Dijkstra which take $\Tilde{O}(m\sqrt{n})$ total time.

The main difference from \cite{CL21} lies in the sparsification rule at \cref{line:checkineq}. Our rule is based on comparing distances in the original input graph $G$, while Chechik and Lifshitz's rule was based on distances in the sparsified graph $G_i$.

\begin{algorithm}[h]

\caption{\label{alg:3-spanner}$\textsc{3-roundtrip-spanner}(G)$}
\KwIn{A weighted directed graph $G = (V,E)$}
\KwOut{a 3-roundtrip spanner $H\subseteq G$}

 $H \gets (V(G), \varnothing)$\\
 $G_0\gets G$\\

 Let $\Delta := \lceil \log_{3/2} \sqrt{n} \rceil$, and $\alpha:= (\sqrt{n})^{1/\Delta}$. \label{line:defndeltaalpha} \tcp{$\alpha\in [5/4,3/2]$ when $\sqrt{n}\ge 2$}

 \For{$i \gets 0,1,\dots, \Delta - 1$} {
  
  Sample $S_i\subseteq V$ by including each vertex with probability $\alpha^i/n$ independently
  \label{line:spannersample}
  \\

  Compute $d_G(s,v), d_G(v,s)$ for all $s\in S_i$ and $v\in V$ using Dijkstra \label{line:dij}\\

  Add to $H$ the shortest path trees in $G$ from/to every vertex in $s\in S_i$ \label{line:add-tree-in-original-graph} \\

  $G_{i+1}\gets G_i$
  \label{line:8}
  \\
  
  \For{$(x,y), (x,s)\in E(G_i)$ such that $s\in S_i$} {
        
        \If{$2d_G(x,s)+d_G(s,y)\le 2\wt(x,y) + d_G(y,s)$ \label{line:checkineq}} 
        {Remove the edge $(x,y)$ from $G_{i+1}$
  \label{line:11}
        }
    
    }

 }

$H\gets H\cup E(G_{\Delta})$ \label{line:addlast}

\Return $H$

\end{algorithm}

\begin{remark}
    \label{remark:remarkresampling}
    Readers familiar with \cite{CL21} may notice some other technical differences between \cref{alg:3-spanner} and \cite{CL21}: in order to remove a $\log n$ factor from the spanner size, Chechik and Lifshitz \cite{CL21} had to resample $S_i$ in case it is ``unsuccessful'' (i.e., \cref{line:11} did not remove sufficiently many edges), whereas our \cref{alg:3-spanner} achieves the same goal  without resampling.  
    Another difference is that we fix the sample rate of each iteration $i$ at \cref{line:spannersample}, while \cite{CL21} determines sample rate based on the current size $|E(G_i)|$.

These modifications are not essential for obtaining this $3$-spanner result. In particular, our algorithm is equivalent to simply sampling $\sum_{i = 0}^{\Delta-1}|S_i| = O(\sqrt{n})$ vertices all at once. Nonetheless, we present it in this way because it leads to cleaner implementation and analysis. Furthermore, it will be useful later for our emulator algorithm in \cref{sec:emulator} (where we require $S_i$ to be uniformly and independently sampled).

\end{remark}

Now we prove the stretch of the spanner constructed by \cref{alg:3-spanner}.
Our proof mostly follows \cite{CL21}; the key difference is that \cite{CL21} estimated $d_H(u\lr v)$ as a whole, while our improvement comes from separately estimating $d_H(u,v)$ and $d_H(v,u)$ and combine them to obtain an upper bound for $d_H(u \lr v)$.

\begin{lemma}
    \label{lem:stretch3}
For any two vertices $u,v\in V$,
\[d_H(u,v)\leq 2d_G(u,v) + d_G(v,u).\]
As a consequence,  $d_H(u\lr v)\leq 3d_G(u \lr v)$ for any $u,v\in V$.
\end{lemma}

\begin{proof}
Let $P$ denote the shortest path from $u$ to $v$ in $G$. If $P$ is completely contained in the final $G_{\Delta}$, then by \cref{line:addlast} clearly $d_H(u,v) = d_G(u,v)$ and we are done. 
For the remaining case, consider any iteration $i$ in which some edge $(x,y)$ of $P$ is removed from $G_{i+1}$ at \cref{line:11}.
By \cref{line:checkineq}, there is a vertex $s\in S_i$ such that 
\[2d_G(x,s)+d_G(s,y)\le 2\wt(x,y)+d_G(y,s),\]
which means
\begin{align}
   d_G(x,s)+d_G(s,y)&\le 2\wt(x,y) + d_G(y,s) -d_G(x,s)\nonumber \\
   & \le 2\wt(x,y) + d_G(y,x). \label{eqn:usv}
\end{align}

Since $H$ contains the shortest path trees in $G$ from $s$ and to $s$ (by \cref{line:add-tree-in-original-graph}), we have
\begin{align*}
    d_H(u, v) &\le d_H(u , s) + d_H(s , v)\\ &= d_G(u , s) + d_G(s , v)\\
    & \le d_G(u,x)+d_G(x,s)+d_G(s,y)+d_G(y,v)\\
    & \le d_G(u,x)+2\wt(x,y)+d_G(y,x)+d_G(y,v) \tag{by \cref{eqn:usv}}.
\end{align*}
Then, using $d_G(y,x)\le d_G(y,v)+d_G(v,u)+d_G(u,x)$, we immediately obtain
\begin{align*}
d_H(u,v) &\le 2\big (d_G(u,x)+\wt(x,y)+d_G(y,v)\big )+d_G(v,u)\\ 
& = 2d_G(u,v)+d_G(v,u). \qedhere
\end{align*}
\end{proof}

\subsection{Analysis of sparsity and running time}

Now we analyze the expected size of $H$ and the running time of \cref{alg:3-spanner}. 
We first prove the following lemma that bounds the expected number of edges in $G_i$. From now on we use $m_i:= |E(G_i)|$. Recall from \cref{line:defndeltaalpha} that $\Delta = \lceil \log_{3/2} \sqrt{n} \rceil$, $\alpha= (\sqrt{n})^{1/\Delta}$, and note that $\alpha\in [5/4,3/2]$ when $\sqrt{n}\ge 2$.

\begin{lemma}\label{lem:expected-sparsity}
    For $i = 0,\dots, \Delta$, we have
    \[\E[m_i]\le   2 n^2/\alpha^i.\]
\end{lemma}

\begin{proof}

In the $i$-th iteration, we sample $S_i \subseteq V$ by including each vertex independently with probability $p_i:= \alpha^i/n$. In the following, we focus on a particular vertex $x\in V$, and let $\deg_i(x)=|N_{G_i}(x)|$ denote the out-degree of $x$ in $G_i$. 

For any two out-neighbors $v_s,v_y\in N_{G_i}(x)$,  we say $v_s$ \emph{eliminates} $v_y$, if the inequality at \cref{line:checkineq} holds for $(s,y):=(v_s,v_y)$. Observe that the inequality at \cref{line:checkineq} is (essentially) symmetric with respect to $y$ and $s$, and one immediately observes that for any two $v,v'\in N_{G_i}(x)$ (possibly $v=v'$), either $v$ eliminates $v'$, or $v'$ eliminates $v$.\footnote{ In more detail, by symmetry we can pick $(s,y):= (v,v')$ or $(v',v)$ to satisfy $2d_G(x,s)+d_G(s,y)\le 2d_G(x,y) + d_G(y,s)$. Then, the inequality at \cref{line:checkineq} holds due to $d_G(x,y ) \le \wt(x,y)$.} Then, \cref{line:11} indicates that, for any $v_s,v_y\in N_{G_i}(x)$, if $v_s\in S_i$ and $v_s$ eliminates $v_y$,  then $v_y \notin N_{G_{i+1}}(x)$. Therefore, $\deg_{i+1}(x)$ is the number of out-neighbors of $x$ that are not eliminated by anyone from $S_i$.

For every $v\in N_{G_i}(x)$, let $e_v$ denote the number of $v'\in N_{G_i}(x)$ that eliminates $v$ (including $v$ itself).
We have 
\begin{equation}
    \label{eqn:bound}
1\le e_v\le \deg_i(x)
\end{equation}
and 
\begin{equation}
    \label{eqn:sum}
\frac{1}{|N_{G_i}(x)|}\sum_{v\in N_{G_i}(x)}e_v = \frac{\deg_i(x)+1}{2}.
\end{equation}
Then, over a uniformly independently sampled set $S_i$ of eliminators, we analyze the expected number of out-neighbors of $x$ that are not eliminated, as follows:\begin{align*}
   \E_{S_i}[\deg_{i+1}(x) \mid G_i] &= \sum_{v\in N_{G_i}(x)} (1-p_i)^{e_v}\\
   & \le \frac{|N_{G_i}(x)|}{2}\cdot  (1-p_i)^{1} + \frac{|N_{G_i}(x)|}{2} \cdot (1-p_i)^{\deg_i(x)} \tag{by convexity of $f(x) = (1-p_i)^x$, and \cref{eqn:bound,eqn:sum}}\\
   & = \frac{\deg_i(x)}{2}\cdot (1-p_i + (1-p_i)^{\deg_i(x)})\\
   & \le \frac{\deg_i(x)}{2}\cdot (1+e^{-p_i\deg_i(x)}).
\end{align*}
Multiplying both sides by $p_{i+1}$, 
\begin{align*}
    \E_{S_i}[p_{i+1}\deg_{i+1}(x)\mid G_i] &\le \frac{p_{i+1}\deg_i(x)}{2}\cdot (1+e^{-p_i\deg_i(x)})\\
    & = \frac{\alpha}{2}(p_{i}\deg_i(x)+p_{i}\deg_i(x)e^{-p_i\deg_i(x)}) \tag{by $p_i= \alpha^i/n$}\\
    & < \frac{\alpha}{2} (p_i\deg_i(x) + 1). 
\end{align*}
Hence,
\begin{align*}
    \E[p_{i+1}\deg_{i+1}(x)]  & \le \frac{\alpha}{2}(\E [p_i\deg_i(x)] + 1).
\end{align*}
Since $0\le p_0\deg_0(x) \le \frac{1}{n}\cdot (n-1)< 1$, by induction we obtain $\E[p_i\deg_i(x)] < \alpha/(2-\alpha)$ for all $i$ (recall $\alpha\le 3/2<2$). Summing over all $x\in V$, we obtain
\[ \E[m_i]  = \frac{1}{2}\sum_{x\in V}\E[\deg_i(x)] \le \frac{n}{2}\cdot \frac{\alpha/(2-\alpha)}{p_i} = \frac{\alpha}{4-2\alpha}n^2/\alpha^{i}< 2n^2/\alpha^i. \qedhere \]
\end{proof}

Now we are ready to present the analysis of the sparsity of $H$ and the running time of our algorithm.

\paragraph{Sparsity}

For each iteration $i$ of \cref{alg:3-spanner}, by definition (\cref{line:spannersample}) we have expected sample size 
\[\E[|S_i|] =  \alpha^i,\]
and we add the shortest path trees from / to every vertex in $s\in S_i$ in $G$, which contain $|S_i|\cdot 2(n-1)$ edges. So summing over all iterations, the number of edges we add in expectation is at most (recall $\alpha\ge 5/4$)
\[\E\lbr{\sum_{i = 0}^{\Delta-1} |S_i|\cdot 2(n-1)} = 2(n-1)\cdot \sum_{i = 0}^{\log_{\alpha}\sqrt{n}\, -1} \alpha^i  = \frac{(2n-1)(\sqrt{n}-1)}{\alpha-1}<8n^{3/2}.\]

In the last step (\cref{line:addlast}), we add all the edges in $G_{\Delta}$ to $H$. By \cref{lem:expected-sparsity}, we have 
\[\E[|E(G_{\Delta})|]\le 2n^2/\alpha^\Delta = 2n^{3/2}.\]
Thus, in expectation we add $O(n^{3/2})$ edges to $H$ in total as desired. 

\paragraph{Running Time}

In each iteration $i$, the bottleneck is at \cref{line:dij} where we run $|S_i|$ instances of Dijkstra on $G$, each taking $O(m + n\log n)$ time. The sparsification steps (\cref{line:8} -- \cref{line:11}) can be implemented in $O(|S_i|\cdot |E(G_i)|) \le O(|S_i|\cdot m)$ time.
So in expectation the total time taken by the algorithm is bounded by 
\[\E\lbr{\sum_{i = 0}^{\Delta-1} |S_i|\cdot O(m+n\log n)} = O(m+n\log n) \sum_{i = 0}^{\log_{\alpha}\sqrt{n}\, -1} \alpha^i =O(m\sqrt{n}+n\sqrt{n}\log n).\]

\section{$(2k-1)$-roundtrip emulator in nearly quadratic time}\label{sec:emulator}

In this section, we give the construction of a $(2k-1)$-roundtrip emulator on $O(kn^{1+1/k})$ edges running in $O(kn^2\log n)$ time for $k\ge 3$ (\cref{thm:emulator-main}). Our algorithm does not work for $k=2$. (For $k = 2$, our $3$-roundtrip spanner algorithm from \cref{sec:3-spanner} has $\tilde O(m\sqrt{n})$ time complexity, which is slower than $\tilde O(n^2)$ for any nontrivial input size $m\gg n^{1.5}$.)

\subsection{Algorithm}

Our algorithm carefully combines ideas from Thorup-Zwick distance oracle \cite{TZ01} and the graph sparsification technique introduced in \cite{CL21}. 
The pseudocode of our algorithm is given in \cref{alg:emulator}. The main body contains $(k-1)\Delta = \Theta(\log n)$ iterations (indexed by $i=r\Delta + t$), divided into $(k-1)$ rounds (indexed by $r\in \{0,\dots,k-2\}$), where each round consists of $\Delta$ iterations (indexed by the inner loop variable $t\in \{0,\dots,\Delta-1\}$).
The $i$-th iteration samples a vertex subset $S_i$, whose expected size $\E[|S_i|]$ gradually increases from $1$ in the $0$-th iteration to $\Theta(n^{(k-1)/k})$ in the last iteration.
In each iteration we run in/out-Dijkstra from every sampled vertex $s\in S_i$ on the current (sparsified) graph $G_i \subseteq G$.
Using the obtained distance information from/to $S_i$,
we not only perform the graph sparsification steps (\cref{line:emulator:14}--\cref{line:emulator:17}) as in \cite{CL21}, but also compute pivots $p_i(u)\in S_i$ and bunches $B_i(u)\subseteq S_i$ used in Thorup and Zwick's algorithm \cite{TZ01} (in the roundtrip metric) and adds edges to the emulator $H$ accordingly (\cref{line:defnpivot} -- \cref{line:addbunch}).
The main complication compared to \cite{TZ01} is that we now have a sequence of (gradually sparsified) graphs $G_i$ involved rather than a single graph $G$, and the pivots $p_i(u)$ are defined using the distances on the current graph $G_i$, while the bunches 
 $B_i(u)$ are defined with respect to the pivot $p_{r\Delta-1}(u)$ on the graph $G_{r\Delta-1}$ from the \emph{previous round} of the outer loop $r$. 

By our parameter setting, we expect each round in the outer loop to roughly decrease  the size of the current graph by a factor of $n^{1/k}$.
After running all $(k-1)$ rounds, we can show the remaining graph $G_{(k-1)\Delta}$ has $O(n^{1+1/k})$ edges in expectation, and we add all of them to the emulator $H$.

\begin{algorithm}
\caption{\label{alg:emulator}$(2k-1)\textsc{-Emulator}(G)$ (for $k\ge 3$)}
\KwIn{A weighted directed graph $G = (V,E)$}
\KwOut{A $(2k-1)$-roundtrip emulator $H$ of $G$}
 $H \gets (V(G), \varnothing)$\\
 $G_{0} \gets G$\\
 Let $\Delta := \lceil \log_{3/2} n^{1/k} \rceil$, and $\alpha:= ({n}^{1/k})^{1/\Delta}$. \label{line:emulator:defndeltaalpha} \tcp{$\alpha\in [5/4,3/2]$ when $n^{1/k}\ge 2$}
 Let $G_{-1}=$ empty graph and $p_{-1}(u) := \bot $ for all $u\in V$.   \label{line:cornercase}
\tcp{$d_{G_{-1}}(u\lr p_{-1}(u)) = +\infty$.}
 \For{$r \gets 0,\dots,k-2$} {

 \For{$t \gets 0,\dots, \Delta-1$} {
 Let $i := r \Delta + t$\\
  Sample $S_{i}\subseteq V$ by including each vertex with probability $\alpha^i/n$ independently \label{line:emulator:sample}\\
  
    Compute $d_{G_i}(s,v),d_{G_i}(v,s)$ for all $s\in S_i$ and $v\in V$ using Dijkstra \label{line:emulator:dij}\\

    Define pivot $p_{i}(u) := \arg \min_{s\in S_{i}} d_{G_i}(u\lr s)$ for all $u\in V$ \label{line:defnpivot}\\
    
    Define bunch $B_i(u) := \{s \in S_{i}: d_{G_i}(u\lr s) < d_{G_{r\Delta-1}}(u \lr p_{r\Delta-1}(u))\}$. \label{line:defnbunch}\\
    
   \For{$u\in V, s\in \{p_i(u)\} \cup B_i(u)$}{
   
    Add edge $(u,s)$ with weight $d_{G_i}(u,s)$ and edge $(s,u)$ with weight $d_{G_i}(s,u)$ to $H$\label{line:addbunch}\\}
    
   $G_{i+1} \gets G_i$ \label{line:emulator:14}\\
   
   \For{$(x,y), (x,s)\in E(G_i)$ such that $s\in S_i$} {
        
        \If{$2d_{G_i}(x,s) + d_{G_i}(s,y)\le  2\wt(x,y) + d_{G_i}(y,s)$ \label{line:emulator-removal}} 
        {Remove the edge $(x,y)$ from $G_{i+1}$\label{line:emulator:17}}
    
    }

 }
 }
$H \gets H \cup G_{(k-1)\Delta}$ \label{line:emulator:lastline}\\
\Return{$H$}\\
\end{algorithm}

\subsection{Analysis of sparsity and running time}
We can without loss of generality assume $k\le \log n$, since otherwise we can run the algorithm for $k=\lfloor \log n\rfloor$ and still satisfy all the requirements.
 Recall from \cref{line:emulator:defndeltaalpha} that $\Delta = \lceil \log_{3/2} n^{1/k} \rceil$, $\alpha:= ({n}^{1/k})^{1/\Delta}$, and note that $\alpha\in [5/4,3/2]$.

\cref{alg:emulator} has $(k-1)\Delta = \log_{\alpha} n^{1-1/k}$ iterations. 
It has a similar structure as our earlier \cref{alg:3-spanner} for $3$-roundtrip spanner (except for the additional \cref{line:defnpivot} -- \cref{line:addbunch} here).
For each iteration $i = r\cdot \Delta + t$ where $r \in\{ 0,\dots, k-2\}$ and $t\in \{ 0,1,\dots, \Delta-1\}$, by \cref{line:emulator:sample} we have \[\E[|S_i|] = \alpha^i.\]
Similar to the analysis of our $3$-roundtrip spanner algorithm, we have the following lemma on the expected number edges $m_i:=|E(G_i)|$.

\begin{lemma}\label{lem:emulator-expected-sparsity}
In \cref{alg:emulator},    for $0\le i\le (k-1)\Delta$ we have
    \[\E[m_i]\le 2n^2/\alpha^i.\]
\end{lemma}
The proof of \cref{lem:emulator-expected-sparsity} is identical to the proof of \cref{lem:expected-sparsity} for \cref{alg:3-spanner}, and is omitted here.
Note that in \cref{alg:emulator}, \cref{line:defnpivot} -- \cref{line:addbunch} do not affect edges of $G_i$, and the remaining part of the algorithm is almost identical to \cref{alg:3-spanner} except that the number of iterations is changed from $\Delta$ to $(k-1)\Delta$ (and $\alpha$ is changed accordingly), and the sparsification rule (\cref{line:emulator-removal}) now depends on distances of $G_i$ instead of $G$.
These modifications do not affect the proof of \cref{lem:expected-sparsity}.

\paragraph{Running Time} 

Over all iterations of the inner \textbf{for} loop, for every $i = 0,\dots, (k-1)\Delta - 1$, the bottleneck is to run $|S_i|$ instances of in/out-Dijkstras on $G_i$ (\cref{line:emulator:dij}), each taking $O(m_i + n\log n)$ time.
The sparsification steps (\cref{line:emulator:14} -- \cref{line:emulator:17}) can be implemented in $O(|S_i|\cdot m_i)$ time.
Thus by \cref{lem:emulator-expected-sparsity}, the expected total running time of our algorithm can be bounded by (note that $S_i$ and $m_i$ are independent random variables) 
\begin{align*}
    \E[\sum_{i=0}^{(k-1)\Delta-1}|S_i|\cdot O(m_i+n\log n)] &\le  \sum_{i=0}^{(k-1)\Delta-1} \alpha^i\cdot O\big (2 n^2/\alpha^i + n\log n\big )\\
&= \sum_{i=0}^{(k-1)\Delta-1} \alpha^i\cdot O\big (2 n^2/\alpha^i \big )\\
    &= O(n^2\cdot (k-1)\Delta ) \\
    & = O(n^2\log n) .
\end{align*}

\paragraph{Sparsity}

Similar to in \cite{TZ01}, we first bound the expected size of the bunches defined in \cref{line:defnbunch}. As mentioned earlier in \cref{remark:remarkresampling}, here we rely on the property that the vertex samples $S_i$ are uniform and independent.

\begin{lemma}\label{lem:bunch-size}
    For each $i = r\cdot \Delta + t$ (where $r\in \{0,\dots, k-2\}, t \in\{ 0,\dots, \Delta-1\}$) and each vertex $u\in V$, we have
    \[\E[|B_i(u)|] = \alpha^{t+1}.\]
\end{lemma}

\begin{proof}
By definition of $B_i$ at \cref{line:defnbunch}, since $G_i\subseteq G_{r\Delta-1}$ and thus $d_{G_{r\Delta-1}}(\cdot,\cdot) \le d_{G_{i}}(\cdot,\cdot)$, we have
\begin{align*}
|B_i(u)| &= |\{s \in S_{i}: d_{G_i}(u\lr s) < d_{G_{r\Delta-1}}(u \lr p_{r\Delta-1}(u))\}|\\
& \le  |\{s \in S_{i}: d_{G_{r\Delta-1}}(u\lr s) < d_{G_{r\Delta-1}}(u \lr p_{r\Delta-1}(u))\}|.
\end{align*}
Sort all $v\in V$ in increasing order of $d_{G_{r\Delta-1}}(u\lr v)$. Then  $p_{r\Delta-1}(u) = \arg\min_{s\in S_{r\Delta-1}} d_{G_{r\Delta-1}}(u\lr s)$ is the first vertex in this ordering that is included in $S_{r\Delta-1}$, and $|B_i(u)|$ is bounded by the number of vertices included in $S_i$ that occur before $p_{r\Delta-1}(u) $ in this ordering. Since $S_{r\Delta-1}$ and $S_{i}$ are sampled uniformly and independently (conditioned on this ordering determined by $G_{r\Delta-1}$), the expected number of vertices included by $B_i(u)$ is at most
\[ \sum_{j=1}^n \frac{\alpha^i}{n}\cdot \Big (1-\frac{\alpha^{r\Delta-1}}{n}\Big )^j \le \frac{\alpha^i/n}{\alpha^{r\Delta-1}/n}=\alpha^{t+1}. \qedhere\]
\end{proof}

As a direct corollary, we can bound the expected total bunch size.

\begin{cor}\label{cor:round-bunch-size}
    \[ \E\lbr {\sum_{i=0}^{(k-1)\Delta-1}\sum_{u\in V}|B_i(u)|} \le O(kn^{1+1/k}). \]
\end{cor}
\begin{proof}
    For each $r \in \{ 0,\dots, k-2\}$, by \cref{lem:bunch-size} and linearity of expectation, we have     \[\E\lbr{\sum_{t = 0}^{\Delta-1} |B_{r \Delta + t}(u)|} = \sum_{t = 0}^{\log_{\alpha}n^{1/k}\, - 1}\alpha^{t+1} = O(n^{1/k})\]
    for each $u\in V$. Summing over all $r\in \{ 0,\dots, k-2\}$ and $u\in V$, 
    \[ \E\lbr {\sum_{i=0}^{(k-1)\Delta-1}\sum_{u\in V}|B_i(u)|}= \sum_{u\in V}\sum_{r = 0}^{k-2}\E\lbr{\sum_{t = 0}^{\Delta-1} |B_{r \Delta+t}(u)|} \le O(kn^{1+1/k}). \qedhere \]
\end{proof}

Now we can analyze the size of the emulator constructed by \cref{alg:emulator}.
\begin{lemma}
    The emulator $H$ returned by \cref{alg:emulator} has expected size
   \[ \E[|H|]  \le O(kn^{1+1/k}).\]
\end{lemma}
\begin{proof}
    By \cref{cor:round-bunch-size}, the total number edges added at \cref{line:addbunch} has expectation at most
    \[ \sum_{i=0}^{(k-1)\Delta-1}\sum_{u\in V}2(|B_i(u)|+1) \le O(kn^{1+1/k}).\]

In the end at \cref{line:emulator:lastline}, we add all the edges in $G_{(k-1)\Delta}$ to $H$. By \cref{lem:emulator-expected-sparsity}, we know that 
\[\E[m_{(k-1)\Delta}]\le  2n^2/\alpha^{(k-1)\Delta}=
2n^2/\alpha^{(k-1)\log_{\alpha}n^{1/k}} =2n^{1+1/k}.\]

Thus the expected size of $H$ is $O(kn^{1+1/k})$ as desired. 
\end{proof}

\subsection{Stretch analysis}  

By construction, it is clear that $d_H(u,v) \ge d_G(u,v)$ for all $u,v\in V$.

From now on we fix a pair of $u,v\in V$ and consider
the shortest cycle $C$ of length $g := d_G(u\lr v)$ containing the vertices $u,v$. We will prove $d_H(u\lr v) \le (2k-1)d_G(u\lr v)$. 

If $C$ is included in the final sparsified graph $G_{(k-1)\Delta}$, then by \cref{line:emulator:lastline} we know $C$ is included in the emulator $H$ and thus $d_H(u\lr v) = d_G(u\lr v)$. Hence, in the following we assume $C\not \subseteq G_{(k-1)\Delta}$, and let $0\le i <(k-1)\Delta$ be the first iteration in which  $C$ is destroyed by the sparsification steps, that is,  $C\subseteq E(G_i)$ but $C\not\subseteq E(G_{i+1})$. 

We first prove the following \cref{lem:edge-delete-stretch} (which is essentially from \cite{CL21}), which shows that when $C$  is destroyed in iteration $i$, it can be $2$-approximated by a cycle going through some sampled vertex in iteration $i$.

\begin{lemma}\label{lem:edge-delete-stretch}
Then there exists some $s\in S_i$ such that 
\begin{equation}\label{eqn:vsus}
    d_{G_i}(v\lr s)\leq 2g, \text{ and } d_{G_i}(u\lr s)\leq 2g.
\end{equation}
\end{lemma}

\begin{proof}
    By definition of $i$, $d_{G_i}(u\lr v) = g = d_G(u\lr v)$.
Let $(x,y)\in C \setminus E(G_{i+1})$ be an edge on the cycle that is removed. Assume without loss of generality that $(x,y)$ lies on the shortest path from $u$ to $v$ (otherwise, we can swap the roles of $u$ and $v$). By \cref{line:emulator-removal}, this means that there exists some $s\in S_{i}$  where
\begin{equation}
2d_{G_i}(x,s) + d_{G_i}(s,y)\le  2\wt(x,y) + d_{G_i}(y,s),\label{eqn:stretch:temp1}
\end{equation}
which implies the following estimate on the length of the shortest cycle going through $u,s,v$ in $G_i$:
\begin{align*}
       & d_{G_i}(u,s) + d_{G_i}(s, v)+d_{G_i}(v,u) \\
     \le \ & d_{G_i}(u,x)+d_{G_i}(x,s) + d_{G_i}(s,y)+d_{G_i}(y,v) + d_{G_i}(v,u)\tag{triangle inequality}\\
     \le \ & d_{G_i}(u,x) + 2\wt(x,y) + d_{G_i}(y,s) - d_{G_i}(x,s) + d_{G_i}(y,v)+d_{G_i}(v,u) \tag{by \cref{eqn:stretch:temp1}}\\
     \le \ & d_{G_i}(u,x) + 2\wt(x,y) + (d_{G_i}(y,v)+d_{G_i}(v,u)+d_{G_i}(u,x)+d_{G_i}(x,s)) \\
     &- d_{G_i}(x,s) + d_{G_i}(y,v)+d_{G_i}(v,u)\tag{expanding $d_{G_i}(y,s)$ using triangle inequality}\\
     = \ & 2d_{G_i}(u,x) + 2\wt(x,y) + 2d_{G_i}(y,v) + 2d_{G_i}(v,u)\\
     = \ & 2d_{G_i}(u,v) + 2d_{G_i}(v,u) \tag{$(x,y)$ lies on shortest path from $u$ to $v$}\\
     = \ & 2g.
\end{align*}
Thus 
\[d_{G_i}(u\lr s) \le d_{G_i}(u,s) + d_{G_i}(s, v)+d_{G_i}(v,u) \le 2g\]
and the same holds for $d_{G_i}(v\lr s)$ as desired.
\end{proof}

By \cref{lem:edge-delete-stretch} and the definition of the pivots $p_{i}(u) := \arg \min_{s\in S_{i}} d_{G_i}(u\lr s), p_{i}(v) := \arg \min_{s\in S_{i}} d_{G_i}(v\lr s)$ (\cref{line:defnpivot}), we have 
\begin{equation}\label{eqn:vpivupiu}
d_{G_i}(u\lr p_{i}(u))\leq 2g, \text{ and } d_{G_i}(v\lr p_{i}(v))\leq 2g.
\end{equation}

We first consider the case when both $s\in B_i(u)$ and $s\in B_i(v)$ hold (where $s$ is defined in \cref{lem:edge-delete-stretch}). In this case, we have 
\begin{align*}
    d_H(u\lr v)  & \le d_H(u\lr s) + d_H(s\lr v)\\
    & \le d_{G_i}(u\lr s) + d_{G_i}(s\lr v)
\tag{by \cref{line:addbunch}}
    \\  & \le 4g \tag{by \cref{eqn:vsus}}
    \\ & \le (2k-1) g \tag{since $k\ge 3$}
\end{align*}
as desired. 

Hence it remains to consider the case when either $s\notin B_i(u)$ or $s\notin B_i(v)$. In the following we only consider $s\notin B_i(v)$, and the other case where $s\notin B_i(u)$ follows from an analogous argument.

By definition of bunches at \cref{line:defnbunch}, $s\notin B_i(v)$ implies 
\begin{equation}
    \label{eqn:temp1}
    d_{G_i}(v\lr s) \ge d_{G_{r\Delta-1}}(v\lr p_{r\Delta-1}(v)),
\end{equation}
where $i = r\Delta+t$ ($r\in \{0,\dots,k-2\}, t\in \{0,\dots,\Delta-1\}$).
Now we use an induction similar to \cite{TZ01}.

\begin{lemma}
    \label{lem:induct1}
    Suppose integer $J\ge 0$  satisfies
    \begin{itemize}
        \item $p_{(r-j)\Delta-1}(v) \notin B_{(r-j)\Delta-1}(u)$ for all even $0\le j< J$, and
        \item $p_{(r-j)\Delta-1}(u) \notin B_{(r-j)\Delta-1}(v)$ for all odd $0\le j< J$.
    \end{itemize}
    Then, 
    \begin{itemize}
        \item If $J$ is even, then \[ d_{(r-J)\Delta-1}(v \lr p_{(r-J)\Delta-1}(v)) \le (J+2)g.\]
        \item If $J$ is odd, then \[ d_{(r-J)\Delta-1}(u \lr p_{(r-J)\Delta-1}(u)) \le (J+2)g.\]
    \end{itemize}
\end{lemma}
\begin{proof}
    We prove by induction on $J$.
   The base case $J=0$ follows from  
   \begin{align*}
    d_{G_{r\Delta-1}}(v\lr p_{r\Delta-1}(v)) 
    & \le d_{G_i}(v \lr s)\tag{by \cref{eqn:temp1}}\\
    & \le 2g.\tag{by \cref{eqn:vsus}}
   \end{align*}
To prove the inductive case $J\ge 1$, we first consider the case with odd $J$.
By the assumption for $j=J-1$, we have
 $p_{(r-J+1)\Delta-1}(v) \notin B_{(r-J+1)\Delta-1}(u)$. By definition of bunches at \cref{line:defnbunch} (at iteration $i=(r-J+1)\Delta-1 = (r-J)\Delta + (\Delta-1)$), this means
\begin{equation}
    \label{eqn:temp2}
    d_{G_{(r-J+1)\Delta-1}}(u \lr p_{(r-J+1)\Delta-1}(v)) \ge d_{G_{(r-J)\Delta-1}}(u \lr p_{(r-J)\Delta-1}(u)).
\end{equation}
Then, 
\begin{align*}
  d_{G_{(r-J)\Delta-1}}(u \lr p_{(r-J)\Delta-1}(u)) 
& \le d_{G_{(r-J+1)\Delta-1}}(u \lr p_{(r-J+1)\Delta-1}(v)) \tag{by \cref{eqn:temp2}}\\
&\le d_{G_{(r-J+1)\Delta-1}}(v \lr p_{(r-J+1)\Delta-1}(v)) + d_{G_{(r-J+1)\Delta-1}}(v \lr u)\tag{triangle inequality}\\
& \le (J-1+2)g + d_{G_{(r-J+1)\Delta-1}}(v \lr u) \tag{by induction hypothesis}\\
& \le (J-1+2)g + g \tag{since $C \subseteq E(G_i) \subseteq E(G_{(r-J+1)\Delta-1})$}\\
& = (J+2)g,
\end{align*}
as desired.

The inductive proof for even $J$ is similar, by switching the role of $u$ and $v$.
\end{proof}

\begin{lemma}
    Let $J\ge 0$ be the maximum integer for which the assumption in \cref{lem:induct1} holds. Then, $d_{H}(u\lr v) \le (2J+5)g$.
    \label{lem:maxJ}
\end{lemma}
\begin{proof}
    We prove the case where $J$ is odd. The even case can be proved similarly by switching the role of $u$ and $v$.

    By the maximality of $J$, we have 
    \begin{equation}
        \label{eqn:pinb}
        p_{(r-J)\Delta-1}(u) \in B_{(r-J)\Delta-1}(v).
    \end{equation}
    By the conclusion of \cref{lem:induct1}, we have
    \begin{equation}
        \label{eqn:concl}
         d_{G_{(r-J)\Delta-1}}(u \lr p_{(r-J)\Delta-1}(u)) \le (J+2)g.
    \end{equation}
        
    Then,
\begin{align*}
    d_H(u\lr v) &\le d_H(u \lr p_{(r-J)\Delta-1}(u)) + d_H(p_{(r-J)\Delta-1}(u) \lr v)\\
    & \le d_{G_{(r-J)\Delta-1}}(u\lr p_{(r-J)\Delta-1}(u)) + d_H(p_{(r-J)\Delta-1}(u) \lr v)\tag{by \cref{line:addbunch}}\\
    & \le d_{G_{(r-J)\Delta-1}}(u\lr p_{(r-J)\Delta-1}(u)) + d_{G_{(r-J)\Delta-1}}(p_{(r-J)\Delta-1}(u) \lr v)\tag{by \cref{line:addbunch} and \cref{eqn:pinb}}\\
    & \le 2d_{G_{(r-J)\Delta-1}}(u\lr p_{(r-J)\Delta-1}(u)) + d_{G_{(r-J)\Delta-1}}(u \lr v)\tag{triangle inequality}\\
    & \le 2(J+2)g +d_{G_{(r-J)\Delta-1}}(u \lr v) \tag{by \cref{eqn:concl}}\\
    & = (2J+5)g.  \qedhere
\end{align*}
\end{proof}

\begin{lemma}
   $d_{H}(u\lr v) \le (2k-1)g.$ 
\end{lemma}
\begin{proof}
   By \cref{line:cornercase} and \cref{line:defnbunch}, we know $B_{\Delta-1}(u) = B_{\Delta-1}(v) = S_{\Delta-1}$. In particular, this means $J$ cannot satisfy the assumption of \cref{lem:induct1} if $J\ge r$.

Hence, the maximum $J$ that could possibly satisfy the assumption of \cref{lem:induct1} is at most $r-1 \le k-3$. Then, by \cref{lem:maxJ}, we have $d_{H}(u\lr v) \le (2J+5)g \le (2(k-3)+5)g = (2k-1)g$. 
\end{proof}

\section{$4$-Approximation of girth in $\Tilde{O}(mn^{1/3})$ time}\label{sec:4-approx}

In this section, we present our algorithm for computing a $4$-approximation of the girth in a weighted directed graph (\cref{thm:4girth-main}).

 In general, we follow the approach of Chechik and Lifshitz \cite{CL21} which uses uniformly random vertex samples and certain elimination rules to prune the search space for each vertex $v\in V$. Our running time improvement comes from extending the framework of \cite{CL21} by one more layer, using several novel structural and algorithmic ideas.
 
Throughout this section, $d(u,v)$ always means $d_G(u,v)$, where $G=(V,E)$ is the input directed graph. 
 
\subsection{Main Algorithm}
By \cref{lem:regularize}, we assume each vertex in $G$ has degree at most $O(m/n)$.

Before describing our algorithm in detail, we first give a high-level overview of the structure of our algorithm. Our algorithm runs in three phases:

\begin{enumerate}
    \item \textbf{Phase I.} Take a random sample $S_1 \subseteq V$ of $O(n^{1/3})$ vertices. 
    
    For each $s_1\in S_1$ use Dijkstra to find the shortest cycle going through $s_1$. 

    \item \textbf{Phase II.} Take a sample $S_2\subseteq V$ of $O(n^{2/3})$ vertices.  
   Based on the distance information from $S_1$ obtained in Phase I, for every $s_2\in S_2$ we use the elimination rule from \cite{CL21} (\cref{lem:observation}) to compute the pruned sets $\Btwoout(s_2),\Btwoin(s_2)\subseteq V$ of size $\Tilde{O}(n^{2/3})$.
    
    For each $s_2\in S_2$ use Dijkstra to find the shortest cycle going through $s_2$ and some $u\in \Btwoout(s_2)\cap \Btwoin(s_2)$.

    \item \textbf{Phase III.} Based on the distance information obtained from Phase I and II,  use our novel elimination rules (\cref{def:bprime} and  \cref{def:bprimetilde}, which are more technical than \cite{CL21}) to compute for every vertex $v\in V$ a pruned set $\tilde B'(v)\subseteq V$ of size $\tilde O(n^{1/3})$.
    
    For each $v\in V$ use Dijkstra to find the shortest cycle going through $v$ in the induced subgraph $G[\tilde B'(v)]$.
\end{enumerate}
Finally output the length of the shortest cycle encountered in the three phases as the girth estimate.

We present our main algorithm in \cref{alg:4approx} as follows. It follows the three-phase structure described above (indicated by the comments), but involves more definitions and subroutines that will be explained in the following sections. The main statements for the correctness and running time of \cref{alg:4approx} will be given in \cref{thm:main4apx} and \cref{thm:main4runtime}.

\begin{algorithm}[H]
\caption{\label{alg:4approx}\textsc{4-Approximation-Girth}$(G)$}
\SetAlgoVlined\DontPrintSemicolon
\KwIn{
A strongly connected directed graph $G=(V,E)$ with maximum degree $O(m/n)$ 
}
\KwOut{An estimate $g'$ such that $g\leq g' \leq 4g$, where $g$ is the girth of $G$}
\BlankLine

Initialize $g'\gets \infty$

\tcp{Phase I}

Sample  $S_1\subseteq V$ of size $|S_1|=O(n^{1/3})$ \label{line:phaseibegin}

    \For {$s_1 \in S_1$} 
        {From $s_1$ run in- and out-Dijkstra  on $G$        

 $g'\gets \min_{u\in V\setminus \{s_1\}} d(s_1\lr u)$ \label{line:update_R1}
        }    

 \tcp{Phase II}

Compute eliminators $\Rout(v),\Rin(v)\subseteq S_1$ of size $|\Rout(v)|,|\Rin(v)| = O(\log n)$ for all $v\in V$  using \cref{alg:compute-eliminator} \label{line:compute_R1}

Sample $S_2\subseteq V$ of size $|S_2|=O(n^{2/3})$ \label{line:samples2}

    \For {$s_2 \in S_2$\label{line:fors2}} 
        {From $s_2$ run modified out-Dijkstra on $G[\Btwoout(s_2)]$ and modified in-Dijkstra on $G[\Btwoin(s_2)]$ (\cref{lem:compute_ball}), where $\Btwoout(\cdot),\Btwoin(\cdot)$ are defined in \cref{def:balls} \tcp{$\Btwoout(\cdot)$ and $\Btwoin(\cdot)$ depend on $\Rout$ and $\Rin$ respectively.}
  \label{line:update_R2}

$g' \gets \min \{g', \min_{u\in \Btwoout(s_2)\cap \Btwoin(s_2) \setminus \{s_2\}} \big (d(s_2 , u)+d(u,s_2)\big )\}$ 
        \label{line:dij2}
        }

 \tcp{Phase III}

 Compute eliminators $\Rtwoin(v)\subseteq S_2$ of size $|\Rtwoin(v)|= O(\log n)$ for all $v\in V$ using \cref{alg:compute-eliminator2}. \label{line:compute_Rtwo}

    \For {$v \in V$} 
        {From $v$ run modified in-Dijkstra  on $G[\tilde B'(v)]$, where $\tilde B'(v)$ is defined in \cref{def:bprimetilde} \tcp{ $\tilde B'(\cdot )$ depends on $\Rtwoin$ (and also $\Rout$).}\label{line:finaldij1}
        
        $g'\gets \min\{ g', \min_{u\in G[\tilde B'(v)]\text{ and } (v,u)\in E} \big (d_{G[\tilde B'(v)]}(u,v)+\wt(v,u) \big ) \}$
        \label{line:finaldij2}
        }    
   
    \Return $g'$
\end{algorithm}

\subsection{Phase I and II}\label{subsec:phase2}

In this subsection we describe Phase I and II of our \cref{alg:4approx}, which mostly follow the $2$-approximation algorithm of \cite{CL21} (with sample size $|S_1|$ changed from $O(\sqrt{n})$ to $O(n^{1/3})$). One piece missing from \cite{CL21} but necessary for us is a certain closedness property of the pruned sets $\Btwoout(v)$, which allows us to find all vertices in $\Btwoout(v)$ by simply running Dijkstra from $v$ (\cref{lem:compute_ball}).\footnote{We need to compute these pruned sets $\Btwoout(v)$ in order to prepare for the later Phase III, which was not required in \cite{CL21}'s two-phase algorithm.}

Phase I (\cref{line:phaseibegin}--\cref{line:update_R1}) uniformly samples a set $S_1$ of $O(n^{1/3})$ vertices, and runs $O(n^{1/3})$ Dijkstra instances on $G$ to find the shortest cycle going through any vertex in $S_1$.
\begin{observation}
    \label{obs:phase1runtime}
Phase I of \cref{alg:4approx} runs in $\Tilde{O}(mn^{1/3})$ total time.
\end{observation}

In Phase II we try to find other short cycles in $G$ that are not $2$-approximated by the estimate 
obtained in  Phase~I. The first step (\cref{line:compute_R1}) computes eliminators $\Rout(v),\Rin(v)\subseteq S_1$ of small size $|\Rout(v)|,|\Rin(v)|\le O(\log n)$ for all $v\in V$. 
Intuitively these eliminators retain the usefulness of the sample $S_1$ in effectively pruning the search space, while being small enough for the benefit of time efficiency.
We defer the algorithm for computing eliminators (\cref{alg:compute-eliminator}) to the end of this subsection; instead we first present the following important definition that relies on these eliminators $\Rout(v),\Rin(v)\subseteq S_1$.

\begin{definition}[$\Btwoout(v)$ and $\Btwoin(v)$, \cite{CL21}] For $v\in V$, given $\Rout (v),\Rin (v)\subseteq V$, we define vertex subsets
\label{def:balls}
\[\Btwoout(v) = \{u\in V: 2d(v,r_1)+d(r_1,u)> 2d(v,u)+d(u,r_1)\text{ for all }r_1\in \Rout (v)\}, \]
and symmetrically,
\[\Btwoin(v) = \{u\in V: 2d(r_1,v)+d(u,r_1)> 2d(u,v)+d(r_1,u)\text{ for all }r_1\in \Rin (v)\}. \]
\end{definition}

\cref{def:balls} is motivated by the following lemma, which follows from the key observation (\cref{lem:observation}) of \cite{CL21}.
Intuitively it says $\Btwoout(\cdot)$ captures cycles that cannot be 2-approximated by the estimate in phase I.
\footnote{We use superscript $(2)$ in the notation of $\Btwoout(v)$ for this reason, to distinguish it from the set $\Bfourout(v)$ that will be introduced later in \cref{subsec:4-approx-lemmas}.}

\begin{lemma}[2-approximation \cite{CL21}]
\label{lem:2-approx}
   If $u\notin \Btwoout(v)$, then there exists $r_1\in \Rout(v)\subseteq S_1$ such that $d(r_1 \lr u) \le 2d(u\lr v)$.
   
   The same statement holds if we replace ``$\mathrm{out}$'' by ``$\mathrm{in}$''.
\end{lemma}

\begin{proof}
    By \cref{def:balls}, since $u\notin \Btwoout(v)$, there exists $r_1\in \Rout(v)$ such that 
    \[ 2d(v,r_1)+d(r_1,u)\le 2d(v,u)+ d(u,r_1).\]
    Then applying \cref{lem:observation} to $u,v,r_1$, we have $d(r_1\lr u ) \le 2d(u\lr v)$.
    
    The statement with ``out'' replaced by ``in'' can be proved symmetrically by reversing the edge directions.
\end{proof}

The following corollary of \cref{lem:2-approx} shows that cycles passing through some $s_2\in S_2$ are $2$-approximated by Phase I and II of \cref{alg:4approx}. This is essentially  how \cite{CL21} obtained their $2$-approximation.
\begin{cor}[\cite{CL21}]
    \label{cor:s22aprox}
  Let $s_2\in S_2$ and $C$ be the shortest cycle in $G$ going through $s_2$. Then,  the girth estimate $g'$ obtained by the end of Phase II of \cref{alg:4approx} satisfies $g'\le 2g$, where $g$ denotes the length of $C$.
\end{cor}
\begin{proof}
    We can assume $S_1\cap C = \varnothing$, since otherwise the Phase I of \cref{alg:4approx} can find $C$ and hence $g'\le g$.

   If $C\not \subseteq \Btwoout(s_2)$,  let $u\in C\setminus \Btwoout(s_2)$. Then by \cref{lem:2-approx} there exists $r_1\in S_1$ such that $d(r_1\lr u)\le 2d(u\lr s_2) =2g$, so Phase I of \cref{alg:4approx} will update $g'$ with $d(r_1\lr u)\le 2g$ (since $r_1\in S_1$ and $u\neq r_1$).
   
   Similarly, if $C\not \subseteq \Btwoin(s_2)$, we also have $g'\le 2g$.
   
   The remaining case is $C\subseteq \Btwoout(s_2)\cap \Btwoin(s_2)$. Then, \cref{line:dij2} in Phase II of \cref{alg:4approx} updates $g'$ with $g$.
\end{proof}

The following key lemma (which will be proved later) states that 
the sets $\Btwoin(v),\Btwoout(v)$ defined  using $\Rin(v)$, $\Rout(v)$ returned by \textsc{compute-eliminators-1}$(G,S_1)$ (\cref{alg:compute-eliminator}) have small sizes. 
Intuitively, this is due to the symmetry of the elimination rule in \cref{def:balls} and the sample size being $|S_1|=O(n^{1/3})$.

\begin{lemma}[Sizes of $\Btwoin(v),\Btwoout(v)$, \cite{CL21}]\label{lem:Btwo-size}
    With high probability\footnote{We use ``with high probability'' to mean probability $1-1/n^c$ for arbitrary given constant $c\ge 1$.} over the random sample $S_1 \subseteq V$,
we  have $|\Btwoin(v)|,|\Btwoout(v)|\le \Tilde{O}(n^{2/3})$ for all $v\in V$.
\end{lemma}

Then, the next step of Phase II is to uniformly sample a set $S_2$ of $O(n^{2/3})$ vertices (\cref{line:samples2}). We then run out-Dijkstra from every $s_2\in S_2$ on the induced subgraph $G[\Btwoout(s_2)]$, and update the girth estimate $g'$ with the found cycles going through $s_2$ (\cref{line:fors2}--\cref{line:dij2}). 

In order to implement the out-Dijkstra on $G[\Btwoout(s_2)]$ at \cref{line:update_R2}, we need the following 
\cref{lem:compute_ball} which states that the set $\Btwoout(s_2)$ (as well as distances $d(s_2,u)$ for all $u\in \Btwoout(s_2)$) can be efficiently computed given the eliminators $\Rout(v)$ due to its special structure. Recall that $d(\cdot,\cdot)$ always denotes distances in the input graph $G$.

\begin{lemma}[Compute $\Btwoout(v)$]\label{lem:compute_ball}
For any vertex $v\in V$, given $\Rout(v)$ of size $O(\log n)$, 
there exists an algorithm running in $\Tilde{O}(\frac{m}{n}\cdot |\Btwoout(v)|)$ time that computes the set $\Btwoout(v)$, and the distances $d(v,u)$ for all $u\in \Btwoout(v)$. 

The same statement holds if we replace ``$\mathrm{out}$'' by `` $\mathrm{in}$'' and replace $d(v,u)$ by $d(u,v)$.

\end{lemma}

\begin{proof}
 We run a modified out-Dijkstra from $v$ on graph $G$, and let $D[u]$ denote the length of the shortest path  from $v$ to $u$ found by this out-Dijkstra.  The modification is that whenever we pop a vertex $u$ from the heap, we relax the out-neighbors of $u$ only if $u$ satisfies 
 \begin{equation}
     \label{eqn:tempineq}
 2d(v,r_1)+d(r_1,u)> 2D[u]+d(u,r_1), \text{ for all }r_1\in \Rout(v).
 \end{equation}
 Comparing \cref{eqn:tempineq} with the definition of $\Btwoout$ (\cref{def:balls}), the difference is that we use $D[u]$ in place of $d(v,u)$.
Note that the other three terms in \cref{eqn:tempineq} are already computed in Phase I because $r_1\in S_1$.

To show the correctness of the modified out-Dijkstra, the key claim is the following closedness property of $\Btwoout(v)$:
\begin{claim}
    \label{claim:temp}
If $u\in \Btwoout(v)$, then for every vertex $x$ on the shortest path from $v$ to $u$ in $G$, it holds that $x\in \Btwoout(v)$.
\end{claim}
\begin{proof}
    For all $r_1\in \Rout(v)$, we have
    \begin{align*}
        2d(v,r_1)+d(r_1,x) & \ge 2d(v,r_1)+d(r_1,u)-d(x,u)
        \tag{by triangle inequality}\\
        &> 2d(v,u)+d(u,r_1) - d(x,u) \tag{by $u\in\Btwoout(v)$}\\
        &= 2  d(v,x)+2 d(x,u) + d(u,r_1) - d(x,u) \tag{by assumption on $x$}\\
        &\ge 2d(v,x)+ d(x,r_1). \tag{by triangle inequality}
    \end{align*}
    Hence, we have $x\in\Btwoout(v)$ by definition.
\end{proof}
By \cref{claim:temp}, it is clear that
our modified out-Dijkstra visits exactly all the vertices $u\in \Btwoout(v)$, and correctly computes distances $D[u]=d(v,u)$ for all $u\in \Btwoout(v)$.

 Since $|\Rout(v)| = O(\log n)$, checking the condition \cref{eqn:tempineq} for all $r_1\in \Rout(v)$  only takes $O(\log n)$ time per vertex $u\in V$. 
 By our assumption that the degree of every vertex is at most $O(\frac{m}{n})$, it follows that the modified Dijkstra runs in time $\Tilde{O}(\frac{m}{n}\cdot |\Btwoout(v)|)$.
\end{proof}

Hence, we observe the following corollary: 
\begin{cor}
    \label{obs:forloops2}
\cref{line:fors2}--\cref{line:dij2} of \cref{alg:4approx} take total time $\tilde O(mn^{1/3})$.
\end{cor}
\begin{proof}
By \cref{lem:compute_ball}, the modified out-Dijkstra from all $s_2\in S_2$ takes total time 
\[\tilde O(\frac{m}{n}\sum_{s_2\in S_2}|\Btwoout(s_2)|)\le \tilde O(\frac{m}{n}\cdot |S_2|\cdot n^{2/3})\le \tilde O(mn^{1/3}),\] where we used $|\Btwoout(s_2)|\le \tilde O(n^{2/3})$ from \cref{lem:Btwo-size}.
The update step at \cref{line:dij2} takes $O(m/n)\cdot |\Btwoout(s_2)|$ time for each $s_2\in S_2$, which also sums up to $\tilde O(mn^{1/3})$.
\end{proof}

\paragraph*{Computing eliminators.} Finally, we describe how to compute the eliminators $\Rout(v),\Rin(v)\subseteq S_1$ (\cref{line:compute_R1} of \cref{alg:4approx}). 
This subroutine is basically the same as in \cite{CL21}, but we present it here using our notation for completeness.
See the pseudocode of \textsc{compute-eliminators-1}$(G,S_1)$ in \cref{alg:compute-eliminator}, which takes the uniform vertex sample $S_1\subseteq V$, and returns $\Rout(v)\subseteq S_1$ for all $v$.
The algorithm for computing $\Rin(v)$ is analogous: we simply run \cref{alg:compute-eliminator} on the graph obtained by reversing the edge orientations of $G$, and we omit the detailed descriptions here.

\begin{algorithm}[H]\SetAlgoVlined\DontPrintSemicolon

    \caption{\textsc{compute-eliminators-1}$(G,S_1)$}\label{alg:compute-eliminator}
    \KwIn{
    The input graph $G=(V,E)$, and $S_1 = \{s_1,s_2,\dots,s_{|S_1|}\}\subseteq V$ of size $|S_1|=O(n^{1/3})$ sampled uniformly and independently (with replacement)
    }
    \KwOut{Sets $\Rout(v)\subseteq S_1$ of size $O(\log n)$ for every vertex $v\in V$}
    
    \BlankLine
    
    $T^{(0)}(v), R^{(0)}(v)\gets \varnothing$ for every $v\in V$
    
    \For{$i \in \{ 1,\dots, k\}$ where $k = 10\log n$}{
    $S^{(i)}\gets$ the next $10n^{1/3}/\log n$ samples from $S_1$ \label{line:e1sample}\\

    Run in- and out-Dijkstra from every $s\in S^{(i)}$ on $G$ \label{line:elim1:dij}
    
    \For{$v\in V$}{
    
    $T^{(i)}(v)\gets \{s\in S^{(i)}\mid \forall t\in R^{(i-1)}(v), 2d(v,s) + d(s,t) < 2d(v,t) + d(t,s)\}$  \label{line:elim1:tv}
    
    \If{$T^{(i)}(v)\ne \varnothing$ \label{line:addbegin}}{
    
    $t\gets $ a random vertex $t\in T^{(i)}(v)$
    
    $R^{(i)}(v)\gets R^{(i-1)}(v)\cup \{t\}$
    
    }
    \Else {
    
    $R^{(i)}(v)\gets R^{(i-1)}(v)$
\label{line:addend}
    
    }
    }
    }
     \Return $\Rout(v)\gets R^{(k)}(v)$ for each $v\in V$ 
\end{algorithm}

\cref{alg:compute-eliminator} runs in $k=10\log n$ iterations. In each iteration, it takes $10n^{1/3}/\log n$ fresh vertex samples (from $S_1$), and runs Dijkstra from them on $G$. Then, based on the obtained distance information, it possibly adds one sampled vertex $t$ to each $\Rout(v)$.
By inspecting  \cref{alg:compute-eliminator}, one immediately observes the following properties. 

\begin{observation}
\label{lem:compute-eliminators}
 \cref{alg:compute-eliminator} runs in time $\Tilde{O}(mn^{1/3})$, and outputs sets $\Rout(v)\subseteq S_1$ for all $v\in V$ of size $|\Rout(v)| = O(\log n)$.
\end{observation}
\begin{proof}
First note that the total number of vertex samples required at \cref{line:e1sample}  is $|S^{(1)}\uplus \dots \uplus S^{(k)}|= k\cdot 10n^{1/3}/\log n = 100n^{1/3}\le |S_1|$.
In each iteration $1\le i\le k$, the algorithm only adds at most one sampled vertex $t\in S_1$ to the set $R^{(i)}(v)$ for each $v\in V$ (\cref{line:addbegin}--\cref{line:addend}), so each output set $\Rout(v) = R^{(k)}(v) \subseteq S_1$ and has size $|\Rout(v)|\le k \le O(\log n)$. 
    
In each iteration, the Dijkstra instances at \cref{line:elim1:dij}  take time $|S^{(i)}|\cdot O(m+n\log n) \le \Tilde{O}(mn^{1/3})$. Then, to compute $T^{(i)}(v)\subseteq S^{(i)}$ at \cref{line:elim1:tv}
for each $v\in V$, we check for every $s\in S^{(i)}$ whether $s\in T^{(i)}(v)$, by simply going over all $t\in R^{(i-1)}(v)$ and checking the condition
$2d(v,s) + d(s,t) < 2d(v,t) + d(t,s)$. Note that all four terms in this inequality have already been computed by the in- and out-Dijkstras since $s\in S^{(i)}$ and $t\in S^{(1)}\cup \dots \cup S^{(i-1)}$. So $T^{(i)}(v)$ can be computed in time $O(|S^{(i)}|\cdot |R^{(i-1)}(v)|) = O((n^{1/3}/\log n )\cdot \log n) = O(n^{1/3})$ for each $v\in V$.
Thus each iteration runs in time $O(mn^{1/3})$ time and over all $k = O(\log n)$ iterations, \cref{alg:compute-eliminator} runs in total time $\Tilde{O}(mn^{1/3})$.
\end{proof}

Now we prove the key \cref{lem:Btwo-size}, which states that \cref{alg:compute-eliminator} guarantees $\Btwoout(v)$ and $\Btwoin(v)$ to have small size with high probability.
\begin{proof}[Proof of \cref{lem:Btwo-size}]
The proof more or less follows from Section 6 in \cite{CL21}. For purpose of the proof, we define the sets 
\[B_i(v) = \{u\in V \mid 2d(v,u)+d(u,r) < 2d(v,r) + d(r,u) \, \forall r\in R^{(i)}(v)\}.\]
Then note that by definition $\Btwoout(v) = \{u\in V \mid 2d(v,u)+d(u,r) < 2d(v,r) + d(r,u) \,\forall r\in \Rout(v) \} = B_k(v)$. We want to show that 
\[\Pr\lbr{|B_k(v)| > n^{2/3}\log n} \le \frac{1}{n^2}.\]
We first show that if $|B_i(v)| > n^{2/3}\log n$, then
\[\E\lbr{|B_i(v)| \mmid |B_{i-1}(v)|} \le \frac{3}{4}|B_{i-1}(v)|.\]
By symmetry \footnote{For more details, refer to the proof of Lemma 3.3 in \cite{CL21}} of the condition $2d(v,u)+d(u,s) < 2d(v,s) + d(s,u)$ with respect to $u\in B_{i-1}(v)$ and $s \in S^{(i)}\cap B_i(v)$, for any pair of vertices $u,u'\in B_{i-1}(v)$, either $u$ can eliminate $u'$ or $u'$ can eliminate $u$. Thus given a random $s\in S^{(i)}\cap B_i(v)$, on expectation $s$ can eliminate half of the vertices in $B_{i-1}(v)$. So conditioned on the event that $S^{(i)}\cap B_i(v)\ne \varnothing$, we have the expected size of $B_i(v)$ is at most half the size of $B_{i-1}(v)$. Specifically we have 
\[\E\lbr{|B_i(v)| \mmid S^{(i)}\cap B_i(v)\ne \varnothing}\le \frac{1}{2}|B_{i-1}(v)|.\]

Now since $|S^{(i)}| = 10n^{1/3}/\log n$ is a uniform random sample, we can compute $\Pr[S^{(i)}\cap B_i(v)= \varnothing]$ as 
\begin{align*}
    \Pr\lbr{S^{(i)}\cap B_i(v)= \varnothing} & = \lpr{1 - \frac{|B_i(v)}{n}}^{10n^{1/3}/\log n} 
    \approx \exp \lpr{ - \frac{|B_i(v)|\cdot 10n^{1/3}}{n\log n}} 
    \\ & \le \lpr{\frac{1}{4}}^{\frac{|B_i(v)|}{n^{2/3}\log n}} \le \frac{1}{4}.
\end{align*}
Thus we have
\begin{align*}
    \E\Big[|B_{i}(v)| \mmid |B_{i-1}(v)|\Big] &= \E\lbr{|B_{i}(v)|\mmid |B_{i-1}(v)|, S^{(i)}\cap B_i(v)\ne \varnothing}\cdot \Pr\lbr{S^{(i)}\cap B_i(v)\ne \varnothing} \\
    &\quad + \E\lbr{|B_{i}(v)|\mmid |B_{i-1}(v)|, S^{(i)}\cap B_i(v)= \varnothing}\cdot \Pr\lbr{S^{(i)}\cap B_i(v)= \varnothing} \\
    &\le \frac{1}{2}|B_{i-1}(v)| + \frac{1}{4}|B_{i-1}(v)| = \frac{3}{4}|B_{i-1}(v)|
\end{align*}
as desired.

Now we can easily finish the proof by applying Markov's inequality.
\[\Pr\lbr{|B_k(v)| > n^{2/3}\log n} \le \frac{\E[|B_k(v)]}{n^{2/3}\log n} \le \frac{\lpr{\frac{3}{4}}^k n}{(n^{2/3}\log n)}\le \lpr{\frac{3}{4}}^k n^{1/3} \le \frac{1}{n^2}.\]
\end{proof}

\begin{proposition}
    \label{prop:phase2runtime}
Phase II of \cref{alg:4approx} runs in $\tilde O(mn^{1/3})$ total time.
\end{proposition}
\begin{proof}
    Follows from \cref{lem:compute-eliminators} and \cref{obs:forloops2}.
\end{proof}

\subsection{New lemmas for $4$-approximation}
\label{subsec:4-approx-lemmas}
In this section we describe our new structural lemmas that are useful for $4$-approximation.

We start with the following \cref{lem:twolayer}, which naturally extends the 2-approximation lemma (\cref{lem:2-approx}) for $\Btwoin(v)$ from one layer to two layers by exploiting the second sample set $S_2$. 
\begin{lemma}
    \label{lem:twolayer}
   Let $u,v\in V$ $(u\neq v)$  and $r_2\in S_2$.  Suppose
   \[ 2d(r_2,v) + d(u,r_2) \le 2d(u,v) + d(r_2,u) .
    \]
    Then, the girth estimate $g'$ obtained by the end of Phase II of \cref{alg:4approx} satisfies $g'\le 4d(u\lr v)$.
\end{lemma}
\begin{proof}
    Apply \cref{lem:observation} (\emph{with edge direction reversed}) to $u,v,r_2$, and obtain  
    \[ d(r_2 \lr u)\le 2d(u\lr v).\]
    If $u\notin \Btwoout(r_2)$, then by \cref{lem:2-approx} there exists $r_1\in \Rout(r_2)\subseteq S_1$ such that \[d(r_1\lr u) \le 2 d(u\lr r_2) \le 4d(u\lr v).\] This implies $g'\le 4d(u\lr v)$ due to the update at \cref{line:update_R1} in Phase I of \cref{alg:4approx} for $r_1\in   S_1$.\footnote{This argument requires $r_1\neq u$. This can be ensured by assuming $u\notin S_1$ without loss of generality: if $u\in S_1$, then Phase I of \cref{alg:4approx} will update $g'$ using $d(u\lr v)$.}

    Similarly, if $u\notin \Btwoin(r_2)$, then we also have $g'\le 4d(u\lr v)$.
    
   It remains to consider the case where $u\in \Btwoin(r_2)\cap \Btwoout(r_2)$. In this case, \cref{line:dij2}  of \cref{alg:4approx} updates $g'$ with $d(r_2 \lr u)\le 2d(u\lr v)$ (here we need to assume $u\neq r_2$; the $u=r_2$ case is already covered by \cref{cor:s22aprox}).
   
   Hence, we always have $g' \le 4d(u \lr v)$.
\end{proof}
In light of \cref{lem:twolayer}, a natural attempt for a 4-approximation algorithm is to imimate Phase II and focus on for each $v\in V$ the pruned vertex set $\{u\in V: 2d(r_2,v)+d(u,r_2) > 2d(u,v) + d(r_2,u)\text{ for all } r_2\in R(v)\}$ for some suitably defined $R(v)\subseteq S_2$. As mentioned in the technical overview, this attempt would require distance information for all $r_2\in S_2$, which is infeasible to compute efficiently enough due to the large size $|S_2|= O(n^{2/3})$. 
Thus, we need to use more structural lemmas for our algorithm, described as follows.

First, we generalize the key observation (\cref{lem:observation}) of \cite{CL21} to the following \cref{lem:k-approx}. Note that \cref{lem:observation} corresponds to the $k=2$ case of \cref{lem:k-approx}.  See \cref{fig:k-approx-lem} (the same figure as \cref{lem:observation}) for an illustration.
 \begin{lemma}[Generalized key observation]
 \label{lem:k-approx}
For any $k\ge 1$ and vertices $u,v,r$, if 
 \[k\cdot d(v,r) + d(r,u) \le k\cdot d(v,u) + (k-1)\cdot d(u,r) ,\] then \[ d(r\lr u)\le k\cdot d(u\lr v).\]
 \end{lemma}
 \begin{proof}
 Note that by triangle inequality, we have $d(u,v) \geq d(u,r) - d(v,r),$ so
 \begin{align*}
     k\cdot d(v,u) + k\cdot d(u,v) & \geq k\cdot d(v,u) + k\cdot d(u,r) - k\cdot d(v,r)  \\
     & \geq \big (k\cdot d(v,r) + d(r,u) - (k-1)\cdot d(u,r)\big ) + k\cdot d(u,r) - k\cdot d(v,r)\\
     & = d(r,u)+d(u,r). \qedhere
 \end{align*}
 \end{proof}
\cref{lem:k-approx} inspires the following definition of $\Bfourout(v)$ and a $4$-approximation lemma (\cref{lem:4-approx}), which are analogous to $\Btwoout(v)$ (\cref{def:balls}) and the $2$-approximation lemma (\cref{lem:2-approx}).
 
\begin{definition}[$\Bfourout(v)$]
    \label{def:4balls}
For $v\in V$, given $\Rout (v)\subseteq V$, we define vertex subsets
\[\Bfourout(v) = \{u\in V: 4d(v,r_1)+d(r_1,u)> 4d(v,u)+3d(u,r_1)\text{ for all }r_1\in \Rout (v)\}. \]
\end{definition}

\begin{cor}[4-approximation]
\label{lem:4-approx}
   If $u\notin \Bfourout(v)$, then there exists $r_1\in \Rout(v)$ such that $d(r_1 \lr u) \le 4d(u\lr v)$.
\end{cor}
\begin{proof}
    By \cref{def:4balls}, since $u\notin \Bfourout(v)$, there exists $r_1\in \Rout(v)$ such that 
    \[ 4d(v,r_1)+d(r_1,u)\le 4d(v,u)+ 3d(u,r_1).\]
    Then applying \cref{lem:k-approx} with $k=4$ to $u,v,r_1$, we have $d(r_1\lr u ) \le 4d(u\lr v)$.
\end{proof}

We also have the following relationship between $\Bfourout(v)$ and $\Btwoout(v)$.
\begin{lemma}
    \label{lem:containb4}
    For all $v\in V$, $\Bfourout(v) \subseteq \Btwoout(v)$. 
    
    As a consequence,  the algorithm of \cref{lem:compute_ball} for computing $\Btwoout(v)$ can also compute $\Bfourout(v)$ in the same running time.
\end{lemma}
\begin{proof}
  If $u\in \Bfourout(v)$, then by \cref{def:4balls} for   all $r_1\in \Rout(v)$,
  \begin{align*}
  2d(v,r_1)+d(r_1,u) &> 4d(v,u)+3d(u,r_1) - 2d(v,r_1) \\
  &= 2d(v,u)+d(u,r_1) + 2\big (d(v,u)+d(u,r_1)-d(v,r_1)\big ) \\
  & \ge 2d(v,u)+d(u,r_1).
  \end{align*}
  So $u\in \Btwoout(v)$ by \cref{def:balls}.
\end{proof}

Now we state and prove our main novel technical lemma, which is a key ingredient of our 4-approximation algorithm.

\begin{lemma}[4-Approximation Filtering Lemma]\label{lem:filter-lemma}
Consider vertices $r_2,v,u\in V$ such that $v\in \Bfourout(r_2)$ and $u \not\in \Btwoout(r_2)$. Then there exists $r_1 \in \Rout(r_2)$ such that $d(v\lr r_1) \le 4d(v\lr u) $.
\end{lemma}

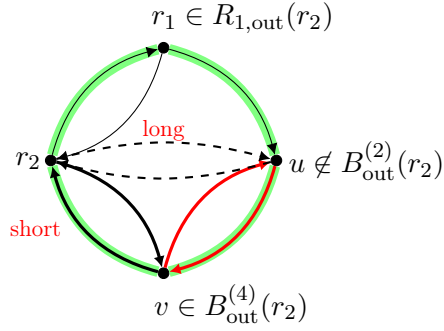
\begin{figure}[H]
\centering
\tikzset{vtx/.style = {circle, draw, fill=black, inner sep=0pt, minimum width=4pt},>={Latex[width=1.5mm,length=1.5mm]}}
    \begin{tikzpicture}
    \node[vtx] (v) at (0,0) {};
    \node[vtx] (u) at (1.5,1.5){};
    \node[vtx] (r2) at (-1.5,1.5){};
    \node[vtx] (r1) at (0,3){};

    \node at (0.9,-0.4) {$v \in \Bfourout(r_2)$};
    \node at (2.7, 1.5) {$u \not\in \Btwoout(r_2)$};
    \node at (-1.8, 1.5) {$r_2$};
    \node at (1.05, 3.4) {$r_1 \in \Rout(r_2)$};

    \draw[green!50, line width=4pt] (v) to[bend left] (r2);
    \draw[green!50, line width=4pt] (r2) to[bend left] (r1);
    \draw[green!50, line width=4pt] (r1) to[bend left] (u);
    \draw[green!50, line width=4pt] (u) to[bend left] (v);

    \draw[->](r1) to[bend left] (r2);
    \draw[->](r2) to[bend left] (r1);

    \draw[->, red, very thick](v) to[bend left] (u);
    \draw[->, red, very thick](u) to[bend left] (v);

    \draw[->, very thick](v) to[bend left] (r2);
    \draw[->, very thick](r2) to[bend left] (v);

    \draw[->,thick, dashed](r2) to[out=15, in=165] (u);
    \draw[->,thick,dashed](u) to[out=-165, in=-15] (r2);
    
    \draw[->](r1) to[bend left] (u);

    \node[scale=0.8,red] at (0,1.9) {long};
    \node[scale=0.8,red] at (-1.7, 0.6) {short};
    \end{tikzpicture}
\caption{
Illustration of the relationship between the vertices involved in \cref{lem:filter-lemma}. The two bold black cycle is relatively short and the dashed cycle is relatively long, the goal is to approximate the red cycle using the cycle highlighted green. As labeled, $v\in \Bfourout(r_2)$ meaning that $v$ and $r_2$ are in a short cycle, $u\not\in \Btwoout(r_2)$ meaning that $u$ and $r_2$ are in a relatively long cycle. Then we can find some $r_1$ in the set of eliminators for $r_2$ such that the cycle passing through $v$ and $r_1$ (highlighted green) approximate the red cycle passing through $v$ and $u$.
}
\label{fig:filter-lem-intuition-main}
\end{figure}

\begin{proof}
Since   $u \not\in \Btwoout(r_2)$, by \cref{def:balls} there exists $r_1 \in \Rout(r_2)$ such that 
\begin{equation}
    2d(r_2,u) + d(u,r_1) \geq 2d(r_2, r_1) + d(r_1, u). \label{eqn:notin2bfs}
\end{equation}

Since $v\in \Bfourout(r_2)$ and $r_1\in \Rout(r_2)$, by \cref{def:4balls} we have
    \begin{equation}
        4d(r_2, r_1) + d(r_1, v) > 4d(r_2,v) + 3d(v,r_1).  \label{eqn:in4bfs}
    \end{equation}
Adding  \cref{eqn:notin2bfs} multiplied by $2$ with \cref{eqn:in4bfs}, and cancelling $4d(r_2,r_1)$ on both sides, we get 
\[ 4d(r_2,u)+2d(u,r_1)+d(r_1,v) > 2d(r_1,u)+4d(r_2,v)+3d(v,r_1).\]
Combining with $4d(r_2,v) + 4d(v,u) \ge 4d(r_2,u)$ (triangle inequality), this implies 
\[ 4d(v,u)+2d(u,r_1)+d(r_1,v) > 2d(r_1,u)+3d(v,r_1).\]
Adding $4d(u,v)$ to both sides gives
\begin{align*}
    4d(u\lr v)+2d(u,r_1)+d(r_1,v) &> \big (2d(r_1,u)+2d(u,v)\big )+\big (2d(u,v)+2d(v,r_1)\big ) +d(v,r_1)\\
    & \ge 2d(r_1,v) + 2d(u,r_1)+d(v,r_1),
\end{align*} 
which immediately simplifies to
\[
    4d(u\lr v)  > d(r_1,v) +d(v,r_1) = d(v\lr r_1). \qedhere
  \]  
\end{proof}

\subsection{Phase III}\label{subsec:phase3}

Now we are ready to describe Phase III, the most technical part of our \cref{alg:4approx}. It has a similar structure as Phase II: we first compute eliminators $\Rtwoin(v)\subseteq S_2$ of size $|\Rtwoin(v)|=O(\log n)$ for all $v\in V$, and then use these eliminators to define pruned vertex sets $\tilde B'(v)$ (which is a subset of $B'(v)\cup\{v\}$ which we will define shortly) to search for short cycles.
In light of the 4-approximation filtering lemma (\cref{lem:filter-lemma}), we will ensure the eliminators satisfy the following property (it will be later shown in \cref{lem:compute-eliminators-2}):
\begin{equation}
    \label{eqn:b4outproperty}
\text{For every $v\in V$ and $r_2 \in \Rtwoin(v)$, we have $v\in \Bfourout(r_2)$.}
\end{equation}
Again, we defer the algorithm for computing the eliminators $\Rtwoin(v)$ to the end of this subsection.

We first make the following technical definition of  pruned vertex sets $B'(v)$, which is directly motivated by the structural lemmas from \cref{subsec:4-approx-lemmas}.
\begin{definition}[$B'(v)$]
\label{def:bprime}
For $v\in V$, let $B'(v)$ denote the set of vertices $s\in V$ that satisfy all the following conditions:
\begin{enumerate}
    \item $v\in \Bfourout(s)$, and
    \label{item:cond1}
    \item  $s\in \Btwoout(r_2)$ for all  $r_2\in \Rtwoin(v)$, and 
    \label{item:cond2}
    \item $2d(s,v) + d(r_2,s)   < 2d(r_2,v) + \underline{d}(s,r_2)$ for all $r_2\in \Rtwoin(v)$. (where $\underline{d}$ is defined in \cref{lem:defestimator})
    \label{item:cond3}
\end{enumerate}
\end{definition}
In this definition, condition~\ref{item:cond1} is motivated by the 4-approximation lemma  (\cref{lem:4-approx}), condition~\ref{item:cond2} is motivated by our 4-approximation filtering lemma (\cref{lem:filter-lemma}) and \cref{eqn:b4outproperty},  and condition~\ref{item:cond3} is motivated by \cref{lem:twolayer}. 
For technical reason, condition~\ref{item:cond3} involves a certain distance underestimate that is easier to compute, defined as follows (readers are encouraged to think of the underestimate as the original distance, and skip this definition at first read): 

\begin{lemma}[Under-estimate of $d(u,r_2)$]
\label{lem:defestimator}
For all $u\in V$ and $r_2\in V$, define $\underline{d}(u,r_2)$ as follows:
\begin{itemize}
    \item \textbf{Case $u\in \Btwoin(r_2)$:} 
    
    Let $\underline{d}(u,r_2):= d(u,r_2)$.
    
    \item \textbf{Case $u\notin \Btwoin(r_2)$:} 

Let\begin{equation}
\underline{d}(u,r_2):=  \frac{1}{2}\min_{r_1\in \Rin(r_2)}(2d(r_1,r_2) + d(u,r_1)-d(r_1,u)) .
\label{eqn:estimator}
\end{equation}
\end{itemize}
Then, $\underline{d}(u,r_2) \le d(u,r_2)$ holds.

\end{lemma}

\begin{proof}
    In order to prove $\underline{d}(u,r_2) \le d(u,r_2)$, it suffices to focus on the second case, $u\notin \Btwoin(r_2)$. 
By definition of $\Btwoin(r_2)$ (\cref{def:balls}), there exists $r_1\in \Rin(r_2)$ such that 
\[2d(r_1,r_2)+d(u,r_1)\le 2d(u,r_2)+ d(r_1,u).\]
This immediately implies $\underline{d}(u,r_2)$ as defined in \cref{eqn:estimator} satisfies $ 2\underline{d}(u,r_2) \le 2d(u,r_2)$.
\end{proof}

The following key lemma (analogous to \cref{lem:Btwo-size} from Phase II) bounds the size of $B'(v)$. 
\begin{lemma}[size of $B'(v)$]\label{lem:compute-elim-2-ball-size}
    With high probability over the random samples $S_1,S_2 \subseteq V$,
we  have $|B'(v)|\le \Tilde{O}(n^{1/3})$ for all $v\in V$.
\end{lemma}
Intuitively, this is due to the symmetry of the elimination rule (Condition~\ref{item:cond3} in \cref{def:bprime} of $B'(v)$), and because the sample size is $|S_2|=O(n^{2/3})$.
We will prove \cref{lem:compute-elim-2-ball-size} later after describing the algorithm computing eliminators $\Rtwoin(v)$.

Our actual algorithm performs a modified in-Dijkstra from every $v\in V$ on the induced subgraph $G[\tilde B'(v)]$ (\cref{line:finaldij1}), where $\tilde B'(v)$ is a slight variant of $B'(v)$, which we shall define shortly. 
The reason for not using $B'(v)$ is because our modified in-Dijkstra algorithm does not know the true distance $d(s,v)$ needed for checking the condition~\ref{item:cond1} and \ref{item:cond3} in the definition of $B'(v)$.\footnote{Note that we introduced the under-estimate $\underline{d}(s,r_2)$ in condition~\ref{item:cond3} of \cref{def:bprime} for the same reason.} Instead, we use the current distance found by the in-Dijkstra to replace $d(s,v)$. The formal definition is as follows (again, readers are encouraged to skip this definition at first read, and think of $\tilde B'(v)$ as the same as $B'(v)$ for intuition):

\begin{definition}[Modified in-Dijkstra and $\tilde B'(v)$]
\label{def:bprimetilde}
   For $v\in V$, consider the following modified in-Dijkstra algorithm starting from $v$ on graph $G$, where we let $D[u]$ denote the length of the shortest path from $u$ to $v$ found by this in-Dijkstra.
   
   The modification is that whenever we pop a vertex $s\neq v$ from the heap, we relax the in-neighbors of $s$ only if $s$ satisfies all the following three conditions:
\begin{enumerate}
    \item $ 4d(s,r_1) + d(r_1,v) > 4D[s] + 3d(v,r_1) \text{ for all } r_1\in \Rout(s),$ and \label{item:cond11}
    
    \item  $s\in \Btwoout(r_2)$ for all  $r_2\in \Rtwoin(v)$, and 
\label{item:cond22}

    \item $2D[s] + d(r_2,s)   < 2d(r_2,v) + \underline{d}(s,r_2)$ for all $r_2\in \Rtwoin(v)$. 
\label{item:cond33}
\end{enumerate}

   Let $\tilde B'(v)$ denote the  set of  vertices $s$ that are popped out from the heap and satisfy all the three conditions above, and additionally we also let $v\in \tilde B'(v)$. (Note that the source vertex $v$ always relaxes all its in-neighbors in the beginning of in-Dijkstra)
\end{definition}
\begin{observation}
    \label{obs:contain}
    $\tilde B'(v)\subseteq B'(v)\cup \{v\}$ for all $v\in V$.
\end{observation}
\begin{proof}
    Note that the three conditions in \cref{def:bprimetilde} are the same as the three conditions in \cref{def:bprime} except that the terms $d(s,v)$ in condition~\ref{item:cond1} and \ref{item:cond3} are replaced by $D[s]$. Since the distance $D[s]$ found by the in-Dijkstra from $v$ must be greater than or equal to the true distance $d(s,v)$, we see that both condition~\ref{item:cond1} and \ref{item:cond3} are strengthened. Hence, $\tilde B'(v)\subseteq B'(v)$.
\end{proof}
Phase III of our algorithm (\cref{line:finaldij1}) is implemented by the modified in-Dijkstra described in \cref{def:bprimetilde}. It remains to show  that we can implement it efficiently. In particular, we need to show that checking the three conditions in \cref{def:bprimetilde} is efficient. We first show that the underestimate $\underline{d}(u,r_2)$ from \cref{lem:defestimator} can be computed efficiently.
\begin{lemma}[Compute $\underline{d}(u,r_2)$]
    \label{lem:computeestimate}
    For $r_2\in V$, assume we know $\Btwoin(r_2)$ and $d(x,r_2)$ for all $x\in \Btwoin(r_2)$. Then $\underline{d}(u,r_2)$ can then be computed for any $u\in V$ in $O(\log n)$ time. 
\end{lemma}
\begin{proof}
According to the definition in \cref{lem:defestimator}, we first check whether $u\in \Btwoin(r_2)$.
 In the first case where $u\in \Btwoin(r_2)$, the answer is $d(u,r_2)$, which we know by assumption. 
In the second case where $u\notin \Btwoin(r_2)$, we need to compute \cref{eqn:estimator} by going over all $O(\log n)$ many $r_1\in \Rtwoin(r_2)$. The expression of \cref{eqn:estimator} only involves distances $d(r_1,\cdot)$ and $d(\cdot,r_1)$ for $r_1\in \Rin(r_2)\subseteq S_1$, which are  already computed in Phase I of \cref{alg:4approx}. So we can compute the answer in $O(\log n)$ time.
\end{proof}

Now we show $\tilde B'(v)$ can be computed efficiently.

\begin{lemma}
\label{lem:indijtime}
The modified in-Dijkstra of \cref{def:bprimetilde} computes $\tilde B'(v)$ in $\tilde O(\frac{m}{n}\cdot |\tilde B'(v)|)$ time.
\end{lemma}
\begin{proof}
   Suppose the modified in-Dijkstra pops vertex $s$ from the heap.
   \begin{itemize}
    \item     The condition \ref{item:cond11} of \cref{def:bprimetilde} can be checked in $O(1)$ time because we already know $d(r_1,\cdot),d(\cdot ,r_1)$ for all $r_1\in  S_1$ from Phase I of \cref{alg:4approx}.
        \item  The condition \ref{item:cond22} can be checked in $O(|\Rtwoin(v)|)\le O(\log n)$ time since we already computed $\Btwoout(r_2)$ for all $r_2\in S_2$ in Phase II of \cref{alg:4approx}.
            \item 
 For condition \ref{item:cond33}, we need to check $2D[s] + d(r_2,s)   < 2d(r_2,v) + \underline{d}(s,r_2)$ for all $r_2\in \Rtwoin(v)$. \begin{itemize}
    \item Since the check for condition~\ref{item:cond22} has passed, we have $s\in \Btwoout(r_2)$. So we know the value of $d(r_2,s)$ from Phase II of \cref{alg:4approx} (note that $r_2\in S_2$).
    \item Since $r_2\in \Rtwoin(v)$, we have $v\in \Bfourout(r_2)\subseteq \Btwoout(r_2)$ by \cref{eqn:b4outproperty}. So  we know the value of $d(r_2,v)$ from Phase II of \cref{alg:4approx}.
        \item We can compute $\underline{d}(s,r_2)$ in $O(\log n)$ time due to \cref{lem:computeestimate} and Phase II of \cref{alg:4approx}.
 \end{itemize}
   \end{itemize}
   Hence, we can check whether $s\in \tilde B'(v)$ in $O(\log^2 n)$ time.
\end{proof}

Now we are ready to prove that our \cref{alg:4approx} achieves $4$-approximation.

\begin{theorem}[Correctness of \cref{alg:4approx}]
    \label{thm:main4apx}
\cref{alg:4approx} returns $g'$ satisfying $g\le g'\le 4g$, where $g$ is the girth of the input directed graph $G$.
\end{theorem}
\begin{proof}
    Let $C$ be the shortest cycle of $G$ with length $g$.
Consider an arbitrary vertex $v$ on $C$.
    If $v\in S_1$, then $C$ is found in Phase I of \cref{alg:4approx} and hence $g'=g$.
 If all vertices on $C$ are contained in $\tilde B'(v)$, then it is eventually found at \cref{line:finaldij2} in \cref{alg:4approx}, and $g'=g$.
Hence, in the following we assume $v\notin S_1$, and there is some vertex $u\in C$ that is not included in $\tilde B'(v)$. We choose $u$ to be the first ancestor of $v$ on the cycle that is not in $\tilde B'(v)$ (in particular, $u$ is a minimizer of $d(u,v)$ among $u\in C \setminus \tilde B'(v)$). Note that $u\neq v$ because $v\in \tilde B'(v)$ by definition.

Let $x\in C$ denote the out-neighbor of $u$ on the cycle $C$. By our definition of $u$, we know the entire shortest path from $x$ to $v$ on $C$ are contained in $\tilde B'(v)$. Then, $x$ must have relaxed its in-neighbor $u$ during the modified in-Dijkstra, which makes $D[u]$ equal to the true distance $d(u,v)$. The fact that $u\notin \tilde B'(v)$ then means some of the three conditions in \cref{def:bprimetilde} is violated for $u$, which then implies $u\notin B'(v)$, as these conditions are equivalent to the three conditions in the definition of $B'(v)$ (\cref{def:bprime}) due to $D[u]=d(u,v)$.

As $u\notin B'(v)$, we  now divide into three cases depending on which condition in \cref{def:bprime} fails for $u$.
\begin{itemize}
    \item 
 Condition~\ref{item:cond1} fails, i.e.,  $v\notin \Bfourout(u)$.
 
  Then by \cref{lem:4-approx}, there exists $r_1\in \Rout(u)$ such that $4d(u\lr v)\ge d(r_1\lr v) \ge g'$ (due to the update at \cref{line:update_R1} for $r_1\in \Rout(u) \subseteq  S_1$ during Phase I of \cref{alg:4approx}; note that $v\neq r_1$ since $v\notin S_1$). 

\item 
 Condition~\ref{item:cond2} fails, i.e., 
$u\notin \Btwoout(r_2)$ for some  $r_2\in \Rtwoin(v)$.

Since $r_2\in \Rtwoin(v)$, by \cref{eqn:b4outproperty} we have $v\in \Bfourout(r_2)$.
Then, by the 4-approximation filtering lemma (\cref{lem:filter-lemma}), there exists $r_1 \in \Rout(r_2)$ such that $4d(v\lr u) \ge d(v\lr r_1) \ge g'$ (due to the update at \cref{line:update_R1} in Phase I of \cref{alg:4approx}).

\item 
Condition~\ref{item:cond3} fails, and
Conditions~\ref{item:cond1},\ref{item:cond2} hold.
This is saying that
there exists $r_2\in \Rtwoin(v)$ such that 
\begin{equation}
    \label{eqn:apxtemp}
    2d(u,v) + d(r_2,u) \ge  2d(r_2,v) + \underline{d}(u,r_2).
\end{equation}
And, we have  $v\in \Bfourout(u)$ (by Condition~\ref{item:cond1}) and $u\in \Btwoout(r_2)$ (by Condition~\ref{item:cond2}).

We further divide into two cases:
\begin{itemize}
    \item 
Case $u\in \Btwoin(r_2)$: 

In this case we have $\underline{d}(u,r_2) = d(u,r_2)$ by \cref{lem:defestimator}. So we apply \cref{lem:twolayer} to \cref{eqn:apxtemp} and obtain  $g'\le 4d(u\lr v) $.

\item 

Case $u\not\in \Btwoin(r_2)$:

 Plugging the definition $\underline{d}(u,r_2) =\min_{r_1\in \Rin(r_2)} \frac{1}{2}(2d(r_1,r_2) + d(u,r_1)-d(r_1,u))$ (from \cref{lem:defestimator}) into \cref{eqn:apxtemp}, we obtain that there exists $r_1\in \Rin(r_2)$ such that 
 \[
 2d(u,v) + d(r_2,u) \ge 2d(r_2,v) + \frac{1}{2}(2d(r_1,r_2) + d(u,r_1)-d(r_1,u)).\]
Multiplying both sides by $2$, and then adding $4d(v,u)-2d(r_2,u)$ to both sides, we get
\begin{align*}
    4d(u,v) + 4d(v,u)&\ge 4d(r_2,v) + 2d(r_1,r_2) + d(u,r_1) - d(r_1,u) + 4d(v,u)- 2d(r_2,u) \\
    &\ge 2d(r_2,v) + 2d(r_1,r_2) + d(u,r_1) - d(r_1,u) + 2d(v,u) \tag{by triangle inequality $d(r_2,u)\le d(r_2,v) + d(v,u)$}\\
    &\ge d(r_2,v) + d(r_1,r_2) + d(u,r_1) + d(v,u) \tag{by triangle inequality $d(r_1,u)\le d(r_1,r_2) + d(r_2,v)+d(v,u)$}\\
    &\ge d(v\lr r_1).
\end{align*}
Hence, $4d(v\lr u)\ge d(v \lr r_1)\ge g'$ (due to the update at \cref{line:update_R1} in Phase I of \cref{alg:4approx}).
\end{itemize}
\end{itemize}
Hence we have established $g' \le 4d(v\lr u)$ in all three cases.
\end{proof}

\paragraph*{Computing eliminators.} Finally, we describe how to compute the eliminators $\Rtwoin(v)\subseteq S_2$ (\cref{line:compute_Rtwo} of \cref{alg:4approx}).
The algorithm has a similar overall structure as the eliminator computation in Phase II (and
 \cite{CL21}) described earlier (\cref{alg:compute-eliminator}). The main idea is to exploit the symmetry in the definition of $B'(v)$ (\cref{def:bprime}), but here it involves more conditions and we need to be slightly more careful to make sure the running time is $\tilde O(mn^{1/3})$.
See the pseudocode of \textsc{compute-eliminators-2}$(G,S_2)$ in \cref{alg:compute-eliminator2}, which takes the uniform vertex sample $S_2\subseteq V$, and returns $\Rtwoin(v)\subseteq S_2$ for all $v$.

\begin{algorithm}[H]\SetAlgoVlined\DontPrintSemicolon

    \caption{\textsc{compute-eliminators-2}$(G,S_2)$}\label{alg:compute-eliminator2}
    
       \KwIn{
    The input graph $G=(V,E)$, and $S_2 = \{s_1,s_2,\dots,s_{|S_2|}\}\subseteq V$ of size $|S_2|=O(n^{2/3})$ sampled uniformly and independently (with replacement)
    }
    \KwOut{The sets $\Rtwoin(v)\subseteq S_2$ of size $O(\log n)$ for every vertex $v\in V$}
    
    \BlankLine

    $T^{(0)}(v), R^{(0)}(v)\gets \varnothing$ for every $v\in V$
    
    \For{$i \in \{ 1,\dots, k\}$ where $k = 10\log n$}{
    $S^{(i)}\gets$ the next  $10n^{2/3}/\log n$ samples from $S_2$. \label{line:e2sample}\\

    \For{$s\in S^{(i)}$} {
     Compute $\Btwoout(s)$, $\Btwoin(s)$ and $\Bfourout(s)$, and distances $d(s,v)$ for all $v\in \Btwoout(s)$, $d(v,s)$ for all $v\in \Btwoin(s)$, using \cref{lem:compute_ball} \label{line:compute_Btwo}
    }
    
    \For{$v\in V$}{
    
   $T^{(i)}(v)\gets \{s\in S^{(i)}\mid v\in \Bfourout(s) \text{ and } \forall t\in R^{(i-1)}(v),  s\in \Btwoout(t) \text{ and } $ $ 2d(s,v) + d(t,s) < 2d(t,v) + \underline{d}(s,t)\}$, where $\underline{d}(\cdot,\cdot)$ is defined in \cref{lem:defestimator}. \label{line:computeT}
    
    \If{$T^{(i)}(v)\ne \varnothing$}{
    
    $t\gets $ a random vertex $t\in T^{(i)}(v)$
    
    $R^{(i)}(v)\gets R^{(i-1)}(v)\cup \{t\}$
    
    }
    \Else {
    
    $R^{(i)}(v)\gets R^{(i-1)}(v)$
    
    }
    }
    }
     \Return $\Rtwoin(v)\gets R^{(k)}(v)$ for each $v\in V$ 
\end{algorithm}

By inspecting \cref{alg:compute-eliminator2}, we observe the following properties (analogous to \cref{lem:compute-eliminators} for \cref{alg:compute-eliminator} from Phase II).

\begin{observation}
   \label{lem:compute-eliminators-2}
   \cref{alg:compute-eliminator2} runs in time $\Tilde{O}(mn^{1/3})$,
   and outputs sets $\Rtwoin(v) \subseteq S_2$ of size $|\Rtwoin(v)| = O(\log n)$ for all $v\in V$. Moreover, for every $v\in V$ and $s \in \Rtwoin(v)$, we have $v\in \Bfourout(s)$.
\end{observation}

\begin{proof}
    First note that the total number of vertex samples required at \cref{line:e2sample}  is $|S^{(1)}\uplus \dots \uplus S^{(k)}|= k\cdot 10n^{2/3}/\log n = 100n^{2/3}\le |S_2|$.
In each iteration $1\le i\le k$, the algorithm only adds at most one sampled vertex $t\in S_2$ to the set $R^{(i)}(v)$ for each $v\in V$, so each output set $\Rtwoin(v) = R^{(k)}(v) \subseteq S_2$ and has size $|\Rtwoin(v)|\le k \le O(\log n)$. 
    
To prove the moreover part, note that by definition of $T^{(i)}(v)$ at \cref{line:computeT}, $v\in \Bfourout(s)$ holds for all $s\in T^{(i)}(v)$ and thus for all $s\in \Rtwoin(v)$.

It remains to bound the running time. In each iteration, \cref{line:compute_Btwo} takes time $\Tilde{O}(\frac{m}{n}\cdot n^{2/3})$ for each $s\in S^{(i)}$ by \cref{lem:compute_ball} (recall that $|\Btwoout(s)|,|\Btwoin(s)|,|\Bfourout(s)|\le \tilde O(n^{2/3})$ by \cref{lem:Btwo-size} and \cref{lem:containb4}), which sums to $\tilde O(n^{2/3}) \cdot \Tilde{O}(\frac{m}{n}\cdot n^{2/3}) = \tilde O(mn^{1/3})$ in total.

To implement    \cref{line:computeT} efficiently, for any given $v\in V$ we want to quickly go over all $s\in S^{(i)}$ such that $v\in \Bfourout(s)$. 
This can be achieved by a preprocessing stage that iterates over $s\in S^{(i)}$ and inserts $s$ to the $v$-th bucket for every $v\in \Bfourout(s)$, in $\Tilde{O}(n^{2/3}\cdot n^{2/3}) = \Tilde{O}(n^{4/3})$ total time. Then, we show that all four terms in the inequality at \cref{line:computeT} are known from the computation at \cref{line:compute_Btwo} or can be computed efficiently from there: $d(s,v)$ is known because $v\in \Bfourout(s), s\in S^{(i)}$, $d(t,s)$ is known because $s\in \Btwoout(t)$ and $t\in S_2$, $d(t,v)$ is known because $v\in \Bfourout(t)$ and $t\in S_2$, and $\underline{d}(s,t)$ can be computed by \cref{lem:computeestimate} in $O(\log n)$ time because we know $\Btwoin(t)$ and $d(x,t)$ for all $x\in \Btwoin(t)$.

Thus overall $k=O(\log n)$ iterations, \cref{alg:compute-eliminator2} takes $\Tilde{O}(mn^{1/3})$ time.
\end{proof}

Now we prove the key \cref{lem:compute-elim-2-ball-size}, which states that \cref{alg:compute-eliminator2} guarantees $B'(v)$ to have small size with high probability.

\begin{proof}[Proof of \cref{lem:compute-elim-2-ball-size}]

The proof is based on symmetry of elimination, which is similar to
 the earlier proof of \cref{lem:Btwo-size}.

Fix the $v\in V$ from \cref{def:bprime}. 
Due to \cref{item:cond1} of \cref{def:bprime}, here we only need to consider vertices from $C_v:=\{s\in V: v\in \Bfourout(s)\}$.
We make the following definition motivated by \cref{item:cond2} and \cref{item:cond3} of \cref{def:bprime}:
for two vertices $s,t\in C_v$, we say $t$ \emph{eliminates} $s$, if $s\notin \Btwoout(t)$ or $2d(s,v) + d(t,s) \ge 2d(t,v) + \underline{d}(s,t)$.
Then, observe that $B'(v)$ consists of exactly the vertices $s\in C_v$ that are not eliminated by any vertex in $\Rtwoin(v)$.

Now we show that for any $s,t\in C_v$, either $s$ eliminates $t$ or $t$ eliminates $s$. 
Suppose to the contrary that  $s$ does not eliminate $t$, and $t$ does not eliminate $s$. Then we have  inequalities
\[ 2d(s,v)+ d(t,s) < 2d(t,v) + \underline{d}(s,t)\le 2d(t,v) + d(s,t)\]
and  \[ 2d(t,v)+ d(s,t) < 2d(s,v) + \underline{d}(t,s) \le 2d(s,v) + d(t,s),\]
which are contradicting each other.

Having proved this symmetry property, the rest of the arguments is the same as in \cref{lem:Btwo-size}, and we omit it here.
\end{proof}

Finally, we can state the time complexity of the entire \cref{alg:4approx}.
\begin{theorem}[Running time of \cref{alg:4approx}]
    \label{thm:main4runtime}
\cref{alg:4approx} runs in $\Tilde{O}(m n^{1/3})$ time with high probability.
\end{theorem}
\begin{proof}
The running time of Phase I  is $\tilde O(mn^{1/3})$ by \cref{obs:phase1runtime}.
The running time of Phase II  is $\tilde O(mn^{1/3})$ by \cref{prop:phase2runtime}.

For Phase III, \cref{line:compute_Rtwo} (computing eliminators $\Rtwoin(v)$ for all $v\in V$) takes $\tilde O(mn^{1/3})$ time by \cref{lem:compute-eliminators-2}. Then, the \textbf{for} loop takes
$\tilde O(\frac{m}{n}\cdot |\tilde B'(v)|)$ time for each $v\in V$. Since $\tilde B'(v) \subseteq B'(v)\cup \{v\}$ (by \cref{obs:contain}) and $|B'(v)|\le \tilde O(n^{1/3})$ (by \cref{lem:compute-elim-2-ball-size}), the total time for this loop is $n\cdot \tilde O(\frac{m}{n}\cdot n^{1/3}) = \tilde O(mn^{1/3})$.

Thus, the overall running time of \cref{alg:4approx} is $\tilde O(mn^{1/3})$.
\end{proof}

\section{Conclusion}\label{sec:conclusion}

We conclude with a few open questions:
\begin{enumerate}
\item Can we compute $3$-roundtrip spanner in $\tilde O(n^2)$ time (or even faster)?
    \item Can we compute $(2k-1)$-approximate roundtrip emulators faster on sparse graphs?
    \item For the $O(mn^{1/k})$-time roundtrip spanner (or directed girth) algorithm of \cite{chechikliu20}, can we improve its $O(k\log k)$ approximation ratio to $O(k)$? Can our technique be combined with the divide-and-conquer techniques of \cite{PachockiRSTW18,chechikliu20,DalirrooyfardW20}?
    
    \item Can we show fine-grained lower bounds for the task of computing roundtrip spanners? In particular, can we rule out $\tilde O(m)$-time algorithms for computing $(2k-1)$-roundtrip spanners of sparsity $O(n^{1+1/k})$?
\end{enumerate}

\printbibliography

\end{document}